\documentclass[aps,prl,twocolumn,10pt]{revtex4-1}

\usepackage{mathrsfs}
\usepackage{bbm}
\usepackage{amsmath}
\usepackage{amssymb}
\usepackage{amsthm}
\usepackage{graphicx}
\usepackage{MnSymbol}
\usepackage{mathtools}
\usepackage[italicdiff]{physics}
\usepackage{chemformula}

\makeatletter
\usepackage{hyperref}
\usepackage{color}
\definecolor{supcol}{RGB}{10,50,180}
\definecolor{eqcol}{RGB}{220,10,100}
\hypersetup{
	colorlinks,
	citecolor=supcol,
	linkcolor=eqcol,
	urlcolor=supcol
}

\allowdisplaybreaks

\newcommand{\mca}{\mathcal}
\newcommand{\mbb}{\mathbb}

\newcommand{\msf}{\mathsf}

\newcommand{\react}[2]{\overset{#1}{\underset{#2}{\rightleftharpoons}}}

\newcommand{\sectionprl}[1]{{\em #1}\/.---}

\begin{document}
\title{Topological Speed Limit}

\author{Tan Van Vu}
\email{tanvu@rk.phys.keio.ac.jp}

\affiliation{Department of Physics, Keio University, 3-14-1 Hiyoshi, Kohoku-ku, Yokohama 223-8522, Japan}

\author{Keiji Saito}
\email{saitoh@rk.phys.keio.ac.jp}

\affiliation{Department of Physics, Keio University, 3-14-1 Hiyoshi, Kohoku-ku, Yokohama 223-8522, Japan}

\date{\today}

\begin{abstract}
Any physical system evolves at a finite speed that is constrained not only by the energetic cost but also by the topological structure of the underlying dynamics. In this Letter, by considering such structural information, we derive a unified topological speed limit for the evolution of physical states using an optimal transport approach. We prove that the minimum time required for changing states is lower bounded by the discrete Wasserstein distance, which encodes the topological information of the system, and the time-averaged velocity. The bound obtained is tight and applicable to a wide range of dynamics, from deterministic to stochastic, and classical to quantum systems. In addition, the bound provides insight into the design principles of the optimal process that attains the maximum speed. We demonstrate the application of our results to chemical reaction networks and interacting many-body quantum systems.
\end{abstract}

\maketitle

\sectionprl{Introduction}Investigating how fast a system can evolve is one of the central problems in classical and quantum mechanics.
In a seminal work by Mandelstam and Tamm \cite{Mandelstam.1945.JP}, a fundamental bound on the operational time required for the transformation between two orthogonal states for unitary dynamics was derived.
Since then, generalizations of the bound for arbitrary states and nonunitary dynamics have been intensively studied \cite{Uhlmann.1992.PLA,Margolus.1998.PD,Campo.2013.PRL,Deffner.2013.PRL,Taddei.2013.PRL,Pires.2016.PRX,Mondal.2016.PLA,Deffner.2017.NJP,Shanahan.2018.PRL,Okuyama.2018.PRL,Ito.2018.PRL,Shiraishi.2018.PRL,Campaioli.2018.PRL,Funo.2019.NJP,GarcaPintos.2019.NJP,Kieu.2019.PRSA,Hu.2020.PRA,Ito.2020.PRX,Nicholson.2020.NP,Vo.2020.PRE,Fogarty.2020.PRL,Ilin.2021.PRA,Sun.2021.PRL,Vu.2021.PRL,Connor.2021.PRA,Shiraishi.2021.PRR,BolonekLaso.2021.Q,Campo.2021.PRL,Delvenne.2021.arxiv,Hamazaki.2022.PRXQ,Hasegawa.2022.arxiv,Vu.2022.arxiv1,Nakajima.2022.arxiv,Salazar.2022.arxiv,Kolchinsky.2022.arxiv}, leading to the notion of speed limits (see Ref.~\cite{Deffner.2017.JPA} for a review).
These speed limits establish the ultimate rate at which a system can evolve to a distinguishable state and have found diverse applications, for example, in quantum control \cite{Caneva.2009.PRL,Deffner.2014.JPB,Campbell.2017.PRL,Funo.2017.PRL}, quantum metrology \cite{Giovannetti.2006.PRL,Beau.2017.PRL}, and thermodynamics of computation \cite{Lloyd.2000.N,Proesmans.2020.PRL,Deffner.2021.EPL,Zhen.2021.PRL,Vu.2022.PRL,Lee.2022.arxiv,Vu.2022.arxiv1}.

Interacting systems generally form topological structures in their dynamics, such as chemical reaction networks that consist of several species (see the schematic in Fig.~\ref{fig:Cover}).
In general, a state represented by a vector $\vb*{x}_t$ evolves over time and is significantly affected by the topology of the dynamics.
For instance, a Markov jump process with dense connectivity may relax toward an equilibrium state faster than one with sparse connectivity.
A many-body system with long-range interactions can change quantum states faster than one with short-range interactions \cite{Eisert.2013.PRL}.
Although speed limits for state transformations have been intensively investigated, the topological nature arising from the network structure in the dynamics has not been fully accounted for.
Note that conventional speed limits, which read $\tau\ge\mca{L}(\vb*{x}_0,\vb*{x}_\tau)/\overline{v}$, employed non-topological metrics $\mca{L}$, such as the Bures angle, trace norm, quantum Fisher information, etc., to quantify the distance between the initial and final states \cite{Deffner.2017.JPA}. 
These metrics are always upper bounded by a constant that does not scale with the size of the system, whereas the dynamics strongly depends on the system size. 
Velocity $\overline{v}$ is determined by the entire dynamics of the system \cite{Campaioli.2019.Q}, and hence it is generally of the order of system size. 
Consequently, conventional speed limits become trivial (i.e., $\tau \ge \mca{L}(\vb*{x}_0,\vb*{x}_\tau)/\overline{v}\to 0$) as the system increases in terms of size \cite{Bukov.2019.PRX}. 
This indicates that in order to derive meaningful bounds, metrics that capture the topological nature and are scalable with system size should be considered.

\begin{figure}[t]
\centering
\includegraphics[width=1\linewidth]{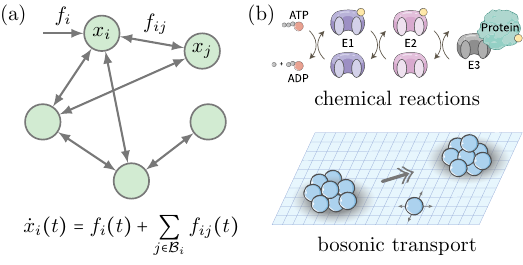}
\protect\caption{(a) Generic time evolution of a physical state $\vb*{x}_t=[x_1(t),\dots,x_N(t)]^\top$ on a graph.
$x_i(t)$ is evolved because of the flows $\{f_{ij}(t)\}$ exchanged between neighboring vertices and an external flow $f_{i}(t)$.
(b) Examples include reactant concentrations in deterministic chemical reactions and boson numbers in interacting bosonic systems.}\label{fig:Cover}
\end{figure}

In this Letter, we derive a speed limit for arbitrary states $\vb*{x}_t$ using a topological metric defined through the network structure in the dynamics.
The time evolution of such states is described by a graph in which each vertex exchanges flows with each other and may be pumped by an external flow.
Examples include the probability distribution of discrete systems, mass concentrations in chemical reaction networks, and vectors of observables in quantum systems (see Fig.~\ref{fig:Cover} again for illustration).
We employ a generalized version of the discrete Wasserstein distance to quantify the distance between the states.
This distance, widely used in optimal transport theory \cite{Villani.2008}, encodes topological information and can grow proportional to the size of the system.
We prove that the minimum operational time required to change the physical state is lower bounded by the Wasserstein distance divided by the average velocity [cf.~Eqs.~\eqref{eq:prim.sl.unnorm} and \eqref{eq:prim.sl.norm}]. The obtained speed limit is tight and can be saturated, even when the system size increases.
Moreover, it is applicable to a broad range of dynamics ranging from deterministic and stochastic classical systems to isolated and open quantum systems.
For example, we apply the theory to chemical reaction networks using the Wasserstein distance applicable to any reversible chemical reaction and provide a reaction speed formula that can discriminate between different chemical reactions \cite{fnt1}.
Another important application is the interacting bosonic transport for arbitrary initial (mixed) states with and without a thermal environment, which is relevant to the Lieb--Robinson velocity \cite{Lieb.1972.CMP}.
Through the examples, we demonstrate that considering topological metrics does not only provide quantitatively tight bounds but also qualitatively reveals the physical mechanism of state transformations, which cannot be obtained with speed limits reported thus far.

\sectionprl{General setup}We consider a time-dependent vector state $\vb*{x}_t\coloneqq[x_1(t),\dots,x_N(t)]^\top$ and an undirected graph ${G}(\mca{V},\mca{E})$ with the vertex set $\mca{V}=\{1,\dots,N\}$ and edge set $\mca{E}$.
Each element $x_i(t)$ corresponds to a vertex $i\in \mca{V}$.
For example, $\vb*{x}_t$ can be a vector of the probability distribution of a discrete system, reactant concentrations of chemical reaction networks, or physical observables in classical and quantum systems (examples are provided later).
For each vertex $i$, let $\mca{B}_i\coloneqq\{j\,|\,\ev{i,j}\in\mca{E}\}$ denote the set of neighboring vertices of $i$.
We assume that the time evolution of $\vb*{x}_t$ is given by the following deterministic equation [see Fig.~\ref{fig:Cover}(a)]:
\begin{equation}
\dot{x}_i(t)=f_i(t)+\sum_{j\in \mca{B}_i}f_{ij}(t),\label{eq:det.eq}
\end{equation}
where $f_{ij}(t)=-f_{ji}(t)$ denotes the flow exchange between vertices $i$ and $j$ for $i\neq j$ and $f_i(t)$ is an arbitrary external flow.
In the absence of external flows [i.e., $f_i(t)=0$ for all $i$], $\sum_{i=1}^Nx_i(t)$ is invariant.
Examples of Eq.~\eqref{eq:det.eq} include the master equation of Markov jump processes, rate equation of chemical reaction networks, and time evolution of the observables in quantum systems.
We define a time-dependent velocity \cite{fnt5}, which is the sum of the absolute values of the external and exchanged flows, given by
\begin{equation}
\upsilon_{t,\lambda}\coloneqq \lambda\sum_i|f_i(t)|+\sum_{\ev{i,j}\in\mca{E}}|f_{ij}(t)|,
\end{equation}
where $\lambda\ge 0$ is a weighting factor, and the second summation is over all unordered pairs $\ev{i,j}\in\mca{E}$.
For simplicity, we denote $\upsilon_{t,0}$ by $\upsilon_{t}$.
We also define the Manhattan norm for an arbitrary vector $\vb*{x}$ as $\|\vb*{x}\|_1\coloneqq\sum_i|x_i|$ and the time average of an arbitrary time-dependent quantity $w_t$ as $\ev{w_t}_\tau\coloneqq\tau^{-1}\int_0^\tau w_t\dd{t}$.

\sectionprl{Wasserstein distance}Here we introduce the discrete $L^1$-Wasserstein distance between two states $\vb*{x}$ and $\vb*{y}$ on the graph ${G}(\mca{V},\mca{E})$.
First, we consider the case in which $\vb*{x}$ and $\vb*{y}$ are balanced (that is, $\sum_ix_i=\sum_iy_i$), and then we generalize the distance to the unbalanced case (that is, $\sum_ix_i\neq\sum_iy_i$).
Let $d_{ij}$ denote the shortest path distance between the vertices $i$ and $j$ in the graph.
In other words, $d_{ij}$ is the minimum length of paths connecting $i$ and $j$.
Graph $G$ is assumed to be connected \cite{fnt2}; therefore, $d_{ij}$ is always finite.
Suppose that we have a transport plan that redistributes $\vb*{x}$ to $\vb*{y}$ by sending an amount of $\pi_{ij}$ from $x_j$ to $y_i$ with a cost of $d_{ij}$ per unit weight for all ordered pairs $\ev{i,j}$.
The Wasserstein distance is then defined as the minimum transport cost for all feasible plans, given by
\begin{equation}
\mca{W}_1(\vb*{x},\vb*{y})\coloneqq\min_{\pi\in\Pi(\vb*{x},\vb*{y})}\sum_{i,j}d_{ij}\pi_{ij}.
\end{equation}
Here, $\Pi(\vb*{x},\vb*{y})$ denotes the set of all transport plans $\pi=[\pi_{ij}]\in\mbb{R}_{\ge 0}^{N\times N}$ that satisfy $\sum_{j}\pi_{ij}=y_i~\text{and}~\sum_{j}\pi_{ji}=x_i$.
Previous studies have shown that the Wasserstein distance plays a crucial role in statistics and machine learning \cite{Kolouri.2017.SPM}, computer vision \cite{Haker.2004.IJCV}, linguistics \cite{Huang.2016.NIPS}, molecular biology \cite{Schiebinger.2019.Cell}, and stochastic thermodynamics \cite{Aurell.2011.PRL,Nakazato.2021.PRR,Vu.2022.arxiv1,Dechant.2022.JPA}.

Next, we describe the generalized Wasserstein distance for the unbalanced case.
Transport between two unbalanced states can be enabled by allowing \emph{add} and \emph{remove} operations in addition to transportation between vertices.
More precisely, an infinitesimal mass $\delta\vb*{x}$ of $\vb*{x}$ can either be removed at cost $\lambda\|\delta\vb*{x}\|_1$ or moved from $\vb*{x}$ to $\vb*{y}$ at cost $\mca{W}_1(\delta\vb*{x},\delta\vb*{y})$.
Mathematically, the generalized Wasserstein distance between unbalanced states can be defined as \cite{Piccoli.2013.ARMA}
\begin{equation}\label{eq:gen.W1}
\begin{aligned}[b]
&\mca{W}_{1,\lambda}(\vb*{x},\vb*{y})\coloneqq\min\qty{\lambda(\|\vb*{x}-\tilde{\vb*{x}}\|_1+\|\vb*{y}-\tilde{\vb*{y}}\|_1)+\mca{W}_1(\tilde{\vb*{x}},\tilde{\vb*{y}})},
\end{aligned}
\end{equation}
where the minimum is over all the states $\tilde{\vb*{x}}$ and $\tilde{\vb*{y}}$ such that $\|\tilde{\vb*{x}}\|_1=\|\tilde{\vb*{y}}\|_1$.
By definition \eqref{eq:gen.W1}, distance $\mca{W}_{1,\lambda}$ always satisfies the triangle inequality \cite{Piccoli.2013.ARMA}.
If $\vb*{x}$ and $\vb*{y}$ are balanced states, then $\mca{W}_{1,\lambda}$ is reduced to $\mca{W}_1$ within the $\lambda\to +\infty$ limit.
We also note that $\mca{W}_{1,\lambda}$ can be calculated numerically using the linear programming method \cite{Supp.PhysRev}.

\sectionprl{Main results}We now utilize the generalized Wasserstein distance \eqref{eq:gen.W1} to derive a topological speed limit for any state $\vb*{x}_t$ obeying the general dynamics \eqref{eq:det.eq}.
Specifically, we prove that the minimum time required to transform $\vb*{x}_0$ into $\vb*{x}_\tau$ is lower bounded by the Wasserstein distance divided by the average velocity:
\begin{equation}\label{eq:prim.sl.unnorm}
\tau\ge\frac{\mca{W}_{1,\lambda}(\vb*{x}_0,\vb*{x}_\tau)}{\ev{\upsilon_{t,\lambda}}_\tau},~\forall \lambda\ge 0.
\end{equation}
In the case that the external flows are absent [i.e., $f_i(t)=0$], inequality \eqref{eq:prim.sl.unnorm} can be reduced to a simple bound by taking the $\lambda\to+\infty$ limit, which reads
\begin{equation}\label{eq:prim.sl.norm}
\tau\ge\frac{\mca{W}_1(\vb*{x}_0,\vb*{x}_\tau)}{\ev{\upsilon_{t}}_\tau}.
\end{equation}
The inequalities \eqref{eq:prim.sl.unnorm} and \eqref{eq:prim.sl.norm} are our main results; the proof is postponed to the end of the Letter.

These results have several physically critical properties.
(i) First, these bounds can be derived as long as the time evolution of $\vb*{x}_t$ is described by Eq.~\eqref{eq:det.eq}, which is a general setting for both the classical and quantum cases.
Notably, the bounds can be saturated if the time evolution \eqref{eq:det.eq} realizes an optimal transport plan.
(ii) Second, our bounds utilize topological information about the system dynamics to provide a stringent constraint on the speed of changing states.
Topological information is encoded into the Wasserstein distance, and this distance term can be as large as the order of the system's size.
(iii) Third, by further upper bounding the time-averaged velocity $\ev{v_{t,\lambda}}_\tau$ by relevant quantities, such as the thermodynamic and kinetic costs, we can derive more interpretable bounds, which clarify the physical mechanism of the speed of state transformations.
(iv) Finally, the speed limit for an arbitrary scalar observable defined in terms of state $\vb*{x}_t$ can also be obtained as a consequence of Eq.~\eqref{eq:prim.sl.unnorm} \cite{Supp.PhysRev}.

In the following, we illustrate the above remarks, especially (i)-(iii), through two applications to classical and quantum systems (see Ref.~\cite{Supp.PhysRev} for further applications in isolated and Markovian open quantum systems, measurement-induced quantum walk \cite{Didi.2022.PRE}, and quantum communication \cite{Bose.2003.PRL,Murphy.2010.PRA}).

\sectionprl{Application 1: Chemical reaction networks}We consider a chemical reaction system composed of several chemical species $X_i$ ($i\in\mca{S}$) that interact through reversible elementary reaction channels $\rho\in\mca{R}$.
Here, $\mca{S}$ and $\mca{R}$ denote the set of indices of the species and reaction channels, respectively.
Each reaction channel is represented as
\begin{equation}
\sum_{i}\nu^{+\rho}_{i}X_i \react{\kappa^{+\rho}}{\kappa^{-\rho}} \sum_{i}\nu^{-\rho}_{i}X_i,
\end{equation}
where $+\rho$ and $-\rho$ correspond to the forward and backward reactions, respectively, $\{\kappa^{\pm\rho}\}$ are the macroscopic reaction rates, and $\{\nu_{i}^{\pm\rho}\}$ are the stoichiometric coefficients.
Let $\vb*{x}_t$ denote the vector of the mass concentrations of species. The molar concentration $\vb*{c}_t$ can be related as $c_i(t)=x_i(t)/m_i$, where $m_i$ denotes the molar mass of species $X_i$.
The time evolution of $\vb*{x}_t$ can be described by the deterministic rate equation:
\begin{equation}\label{eq:rate.eq}
\dot{x}_i(t)=\sum_{\rho}m_i(\nu^{+\rho}_{i}-\nu^{-\rho}_{i})J^{\rho}_t,
\end{equation}
where $J^{\rho}_t\coloneqq J^{-\rho}_t-J^{+\rho}_t$ is the net reaction current and $J^{\pm\rho}_t\coloneqq\kappa^{\pm\rho}\prod_{i}c_i(t)^{\nu_{i}^{\pm\rho}}$ are the reaction fluxes.

Next, we derive the speed limits for the system in terms of the Wasserstein distance defined on graph $G$.
For simplicity, here we consider closed reaction networks, in which the total mass concentration is conserved \cite{fnt6}.
The generalization for open reaction networks, wherein the total mass conservation may be violated, is presented in Ref.~\cite{Supp.PhysRev}.
The total mass conservation law implies $\sum_im_i(\nu_i^{+\rho}-\nu_i^{-\rho})=0$ for any $\rho$.
Due to these conditions, there always exist matrices $\msf{Z}^\rho=[z_{ij}^\rho]$ such that the rate equation \eqref{eq:rate.eq} can be expressed in the form of Eq.~\eqref{eq:det.eq} with $f_{ij}(t)=\sum_\rho z_{ij}^\rho J_t^\rho$ and $f_i(t)=0$ \cite{Supp.PhysRev}.
The graph $G$ can be obtained by adding an undirected edge $\ev{i,j}$ to $\mca{E}$ for any $z_{ij}^\rho\neq 0$.
After some simple manipulations \cite{Supp.PhysRev}, we can prove that 
\begin{equation}\label{eq:vt.ub.crn}
\upsilon_{t}\le\sum_{\rho}\nu^\rho|J^\rho_t|,
\end{equation}
where $\nu^\rho\coloneqq(1/2)\sum_{i}m_i|\nu_i^{+\rho}-\nu_i^{-\rho}|$.
Combining Eqs.~\eqref{eq:prim.sl.norm} and \eqref{eq:vt.ub.crn} yields the following speed limit:
\begin{equation}\label{eq:crn.csl1}
\tau\ge\frac{\mca{W}_1(\vb*{x}_0,\vb*{x}_\tau)}{\ev{\sum_{\rho}\nu^\rho|J^\rho_t|}_\tau}\eqqcolon\tau_1.
\end{equation}
Equation \eqref{eq:crn.csl1} implies that the operational time is lower bounded by the Wasserstein distance and the net reaction currents.

\begin{figure}[b]
\centering
\includegraphics[width=1\linewidth]{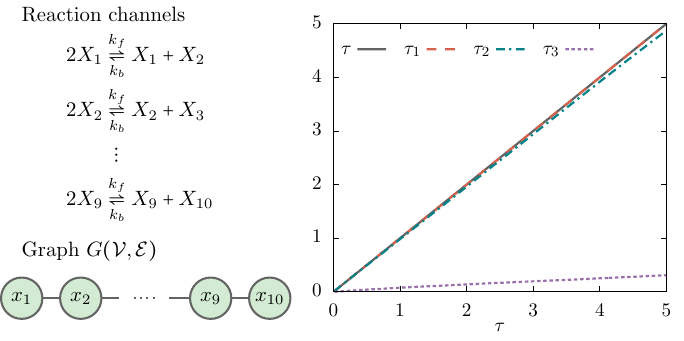}
\protect\caption{Numerical demonstration of the speed limits in the cascade reaction network with $N=10$.
The operational time $\tau$, topological bounds $\tau_1$ and $\tau_2$, and non-topological bound $\tau_3$ are depicted by solid, dashed and dash-dotted, and dotted lines, respectively.
The parameters are set to $k_f=2$ and $k_b=1$.
The initial mass concentration is $\vb*{x}_0=[1,0.9,\dots,0.1]^\top$.}\label{fig:CRNEx}
\end{figure}

A thermodynamic speed limit can also be obtained using Eq.~\eqref{eq:crn.csl1}.
The entropy production rate of a chemical reaction system can be defined as \cite{Rao.2016.PRX}
\begin{equation}
\sigma_t\coloneqq\sum_{\rho}J^\rho_t\ln\frac{J^{-\rho}_t}{J^{+\rho}_t},
\end{equation}
where the gas constant is set to unity. We define the following kinetic quantity:
\begin{equation}
\ell_t\coloneqq\sum_{\rho}(\nu^\rho)^2\frac{J^{-\rho}_t-J^{+\rho}_t}{\ln(J^{-\rho}_t/J^{+\rho}_t)},
\end{equation}
which is the sum of the microscopic Onsager coefficients \cite{Vu.2022.arxiv1,Yoshimura.2022.arxiv}.
Applying the Cauchy--Schwarz inequality, we prove that $\ev{\sum_\rho\nu^\rho|J^\rho_t|}_\tau\le\ev{\sqrt{\sigma_t\ell_t}}_\tau\le\sqrt{\ev{\sigma_t}_\tau\ev{\ell_t}_\tau}$.
Consequently, we obtain the following thermodynamic speed limit:
\begin{equation}\label{eq:crn.csl2}
\tau\ge\frac{\mca{W}_1(\vb*{x}_0,\vb*{x}_\tau)}{\sqrt{\ev{\sigma_t}_\tau\ev{\ell_t}_\tau}}\eqqcolon\tau_2.
\end{equation}
Inequality \eqref{eq:crn.csl2} implies that the minimum time required to transform $\vb*{x}_0$ into $\vb*{x}_\tau$ is determined by the product of the thermodynamic and kinetic costs.

We numerically demonstrate the derived bounds in a cascade reaction network with $|\mca{S}|=10$ species and $|\mca{R}|=9$ reaction channels (see Fig.~\ref{fig:CRNEx}).
We also compare the results with a non-topological bound reported in Ref.~\cite{Yoshimura.2021.PRL}, which reads $\tau\ge\tau_3\coloneqq \mca{T}(\vb*{c}_0,\vb*{c}_\tau)/\sqrt{\ev{\sigma_t}_\tau\ev{d_t}_\tau}$.
Here, $\mca{T}$ denotes the total variation distance and $d_t\coloneqq(|\mca{S}|/8)\sum_{\rho,i}(\nu_i^{+\rho}-\nu_i^{-\rho})^2(J_t^{+\rho}+J_t^{-\rho})$ corresponds to the diffusion coefficient.
We calculate and plot the lower bounds $\tau_i\,(1\le i\le 3)$ in Fig.~\ref{fig:CRNEx}.
As shown, the topological speed limits $\tau\ge\tau_1\ge\tau_2$ are tight; especially, the bound $\tau\ge\tau_1$ is always saturated.
On the contrary, the non-topological bound $\tau\ge\tau_3$ is loose and does not provide a meaningful bound for the speed of the system.

\sectionprl{Application 2: Interacting bosonic systems}Next, we describe an application for quantum many-body bosonic systems.
We consider a model of bosons that hop on an arbitrary finite-dimensional lattice and interact with each other.
Let $\Lambda$ denote the set of all the sites in the lattice.
The Hamiltonian can be expressed in the following generic form:
\begin{equation}
H_t\coloneqq-\gamma\sum_{\ev{i,j}}(b_i^\dagger b_j+b_j^\dagger b_i)+\sum_{Z\subseteq\Lambda}h_Z.
\end{equation}
Here, the first summation is over neighboring lattice sites (which can be arbitrarily distant), $\gamma>0$ describes the boson mobility, $b_i$ and $b_i^\dagger$ are the bosonic creation and annihilation operators for site $i$, respectively, $\hat n_i\coloneqq b_i^\dagger b_i$ is the number operator, and $h_Z$ is an arbitrary function of $\{\hat n_i\}_{i\in Z}$.
Examples include the Bose--Hubbard model, given by $\sum_{Z\subseteq\Lambda}h_Z=(U/2)\sum_i\hat n_i(\hat n_i-1)-\mu\sum_i\hat n_i$, where $U$ and $\mu$ are real constants.
Note that the graph $G(\mca{V},\mca{E})$ of the bosonic system is identical to the lattice topology (i.e., $\mca{V}$ is the set of sites and $\mca{E}$ is the set of edges that connect the two neighboring sites).
The maximum vertex degree of the graph is denoted by $d_G$.

We assume that the bosonic system is weakly coupled to a Markovian thermal reservoir and can exchange particles with the reservoir, where the time evolution of the reduced density matrix is described by the Lindblad equation \cite{Lindblad.1976.CMP}:
\begin{equation}
\dot\varrho_t=-i[H_t,\varrho_t]+\sum_{i\in\Lambda}\qty(\mca{D}[L_{i,+}]+\mca{D}[L_{i,-}])\varrho_t,
\end{equation}
where $\mca{D}[L]\varrho\coloneqq L\varrho L^\dagger-(1/2)\{L^\dagger L,\varrho\}$ is the dissipator, $L_{i,+}=\sqrt{\gamma_{i,+}}b_i^\dagger$ and $L_{i,-}=\sqrt{\gamma_{i,-}}b_i$ are the jump operators that characterize the absorption and emission of bosons at site $i$, respectively.
Hereafter, we set $\hbar=1$ for simplicity.

We consider the vector of boson numbers occupied at each site, $x_i(t)=\tr{\hat n_i \varrho_t}$, and define the instantaneous total number of bosons as $\mca{N}_t\coloneqq\sum_{i\in\Lambda}x_i(t)$.
Using the relation $[b_i,\hat n_i]=b_i$, we can show that the time evolution of $x_i(t)$ can be expressed in the form of Eq.~\eqref{eq:det.eq} with $f_{i}(t)=\tr{L_{i,+}\varrho_tL_{i,+}^\dagger}-\tr{L_{i,-}\varrho_tL_{i,-}^\dagger}$ and $f_{ij}(t)=2\gamma\Im\qty[\tr{b_j^\dagger b_i\varrho_t}]$.
By inserting these terms into $\upsilon_{t,\lambda}$, we can immediately obtain the speed limit \eqref{eq:prim.sl.unnorm} for bosonic transport.

Next, we derive a more physically interpretable speed limit by upper bounding the velocity $\upsilon_{t,\lambda}$. To this end, we introduce two relevant physical quantities.
The first is the irreversible entropy production rate \cite{Landi.2021.RMP}, which is the sum of the entropic changes in the system and environment, defined as $\sigma_t\coloneqq\sigma_t^{\rm sys} + \sigma_t^{\rm env}$.
Here, $\sigma_t^{\rm sys}\coloneqq -\tr{\dot\varrho_t\ln\varrho_t}$ is the rate of von Neumann entropy of the bosonic system, and $\sigma_t^{\rm env}$ quantifies the heat dissipated to the environment as follows:
\begin{equation}
\sigma_t^{\rm env}\coloneqq\sum_i\qty(\tr{L_{i,+}\varrho_tL_{i,+}^\dagger}-\tr{L_{i,-}\varrho_tL_{i,-}^\dagger})\ln\frac{\gamma_{i,+}}{\gamma_{i,-}},
\end{equation}
where we have assumed the local detailed balance condition [that is, $\ln(\gamma_{i,+}/\gamma_{i,-})$ is related to the heat dissipation of the boson exchange at site $i$].
The second is quantum dynamical activity \cite{Hasegawa.2020.PRL,Vu.2022.PRL.TUR}, which quantifies the boson exchange frequency between the system and reservoir, given by
\begin{equation}
a_t\coloneqq\sum_i\qty(\tr{L_{i,+}\varrho_tL_{i,+}^\dagger}+\tr{L_{i,-}\varrho_tL_{i,-}^\dagger}).
\end{equation}
Using these quantities, we can prove that the velocity $\upsilon_{t,\lambda}$ is upper bounded as \cite{Supp.PhysRev}
\begin{equation}\label{eq:vt.ub.boson}
\upsilon_{t,\lambda}\le \gamma d_G\mca{N}_t+\lambda\frac{\sigma_t}{2}\Phi\qty(\frac{\sigma_t}{2a_t})^{-1},
\end{equation}
where $\Phi(x)$ is the inverse function of $x\tanh(x)$.
By combining Eqs.~\eqref{eq:prim.sl.unnorm} and \eqref{eq:vt.ub.boson}, we obtain the following speed limit:
\begin{equation}\label{eq:boson.qsl}
\tau\ge\frac{\mca{W}_{1,\lambda}({\vb*{x}}_0,{\vb*{x}}_\tau)}{\ev{\gamma d_G\mca{N}_t+\lambda{\sigma_t\Phi(\sigma_t/2a_t)^{-1}/2}}_\tau}.
\end{equation}
Equation \eqref{eq:boson.qsl} implies that the speed of bosonic transport is lower bounded by the lattice topology, boson mobility, and dissipation.
The bound also indicates that dissipative controls can help accelerate the bosonic transport.
The inequality \eqref{eq:boson.qsl} is valid for arbitrary initial states of the bosonic system. 

It is worthwhile discussing the vanishing coupling limit (i.e., the case where the system becomes isolated).
In this case, $\sigma_t=a_t=0$ and $\mca{N}_t=\mca{N}$ for all times.
Defining the boson concentration $\bar{x}_i(t)\coloneqq \mca{N}^{-1}x_i(t)$, we obtain $\sum_i\bar{x}_i(t)=1$.
By taking the $\lambda\to+\infty$ limit, Eq.~\eqref{eq:boson.qsl} is reduced to a simple speed limit for an isolated bosonic system:
\begin{equation}\label{eq:boson.qsl2}
\tau\ge\frac{\mca{W}_1(\bar{\vb*{x}}_0,\bar{\vb*{x}}_\tau)}{\gamma d_G}.
\end{equation}
Bound \eqref{eq:boson.qsl2} has a remarkable implication for bosonic transport.
Assume that all bosons are initially concentrated in a region $R_1$, and we want to transport all of them to a distinct region $R_2$ within a finite time $\tau$.
In this case, $\mca{W}_1(\bar{\vb*{x}}_0,\bar{\vb*{x}}_\tau)\ge{\rm dist}(R_1,R_2)$, where ${\rm dist}(R_1,R_2)$ denotes the length of the shortest path connecting the regions $R_1$ and $R_2$.
Therefore, Eq.~\eqref{eq:boson.qsl2} implies that transporting bosons always takes at least a time proportional to the distance between the two regions: $\tau\ge{{\rm dist}(R_1,R_2)}/(\gamma d_G)$, which cannot be obtained with conventional speed limits.
This statement holds for \emph{arbitrary} initial states, including the pure states considered in Ref.~\cite{Faupin.2022.PRL}.
While the Lieb--Robinson bounds \cite{Kuwahara.2021.PRL,Tran.2021.PRL,Faupin.2022.CMP,Yin.2022.PRX,Kuwahara.2022.arxiv} imply a linear light cone for the operator spreading, Eq.~\eqref{eq:boson.qsl2} provides a useful bound for the operational time required for bosonic transport.

\sectionprl{Proof of Eq.~\eqref{eq:prim.sl.unnorm}}We consider the time discretization of Eq.~\eqref{eq:det.eq} with time interval $\delta t=\tau/K$.
For each $k\in[0,K-1]$ and $t=k\delta t$, we have
\begin{equation}\label{eq:dis.det.eq}
x_i(t+\delta t)=x_i(t)+\delta t\qty[ f_i(t) + \sum_{j\in\mca{B}_i}f_{ij}(t)].
\end{equation}
Equation \eqref{eq:dis.det.eq} indicates that we can transform $\vb*{x}_{t}$ to $\vb*{x}_{t+\delta t}$ by adding $f_i(t)\delta t$ to $x_i(t)$ with cost $\lambda|f_i(t)|\delta t$ and exchanging $f_{ij}(t)\delta t$ between neighboring vertices $i$ and $j$ with cost $|f_{ij}(t)|\delta t$.
Such the transport plan takes the total cost of
\begin{equation}\label{eq:cost.short.time}
\qty(\lambda\sum_i|f_i(t)|+\sum_{\ev{i,j}\in\mca{E}}|f_{ij}(t)|)\delta t=\upsilon_{t,\lambda}\delta t,
\end{equation}
which should be larger than or equal to $\mca{W}_{1,\lambda}(\vb*{x}_{t},\vb*{x}_{t+\delta t})$.
Therefore, taking the sum of Eq.~\eqref{eq:cost.short.time} from $k=0$ to $k=K-1$ and applying the triangle inequality for $\mca{W}_{1,\lambda}$ yield
\begin{align}
\sum_{k=0}^{K-1}\upsilon_{t,\lambda}\delta t&\ge\mca{W}_{1,\lambda}(\vb*{x}_{0},\vb*{x}_\tau).\label{eq:short.time.tmp1}
\end{align}
By taking the $\delta t\to 0$ limit in Eq.~\eqref{eq:short.time.tmp1}, we obtain $\tau\ev{\upsilon_{t,\lambda}}_\tau\ge\mca{W}_{1,\lambda}(\vb*{x}_{0},\vb*{x}_\tau)$, from which Eq.~\eqref{eq:prim.sl.unnorm} is immediately derived.

\sectionprl{Conclusion}In this Letter, we derived the topological speed limit for vector states that accounts for the network structure in the underlying dynamics \cite{fnt4}.
The speed limit provides a tight bound for the operational time and insight into the system speed from a topological perspective.
We showed that the bound is applicable to various dynamics as long as the time evolution of the physical state can be described in terms of a graph.
Because our speed limit is derived in a general setting, we expect that it can be applied to obtain fundamental bounds for several other dynamics.

\begin{acknowledgments}
We thank Tomotaka Kuwahara, Andreas Dechant, and Koudai Sugimoto for fruitful discussions, Marius Lemm and Ryusuke Hamazaki for helpful communications, and Mai Dan Nguyen for the help in preparing the figures.
We also appreciate anonymous referees for valuable comments.
This work was supported by Grants-in-Aid for Scientific Research (JP19H05603 and JP19H05791).
\end{acknowledgments}


\begin{thebibliography}{88}
\bibitem [{fnt({\natexlab{c}})}]{fnt2}%
  \BibitemOpen
  \href@noop {} {}A graph is said to be connected if there always exists a path
  from any vertex to any other vertex in the graph. In the case that the graph
  is not connected, we can define the Wasserstein distance as the sum of the
  distances defined for each connected component.\BibitemShut {Stop}%
\end{thebibliography}
\end{document}


\title{Supplemental Material for ``Topological Speed Limit''}

\author{Tan Van Vu}

\affiliation{Department of Physics, Keio University, 3-14-1 Hiyoshi, Kohoku-ku, Yokohama 223-8522, Japan}

\author{Keiji Saito}

\affiliation{Department of Physics, Keio University, 3-14-1 Hiyoshi, Kohoku-ku, Yokohama 223-8522, Japan}

\begin{abstract}
This Supplemental Material describes the details of the analytical calculations presented in the main text and further applications of the topological speed limit for isolated and open quantum systems, measurement-induced quantum walk, and quantum communication. 
The equations and figure numbers are prefixed with S [e.g., Eq.~(S1) or Fig.~S1]. 
The numbers without this prefix [e.g., Eq.~(1) or Fig.~1] refer to the items in the main text.
\end{abstract}

\pacs{}
\maketitle

\tableofcontents

\section{Numerical calculation of the generalized Wasserstein distance}
The Wasserstein distance can be calculated using the linear programming method.
Mathematically, it can be formulated as the following minimization problem:
\begin{equation}\label{eq:lin.prob1}
\begin{array}{ll@{}}
\text{minimize}  & \tr{\msf{C}\pi}\\
\text{subject to}& \pi\vb*{1}=\vb*{y},~\pi^\top\vb*{1}=\vb*{x}\\
                 & \pi\ge 0,
\end{array}
\end{equation}
where $\msf{C}=[d_{mn}]$ is the matrix of transport cost and $\vb*{1}$ is the all-ones vector.

The generalized Wasserstein distance can be calculated in a similar way.
Note that it has been shown that the generalized Wasserstein distance can be achieved using only \emph{remove} and \emph{transport} operations (i.e., we do not need \emph{add} operation) \cite{Piccoli.2013.ARMA}.
Therefore, $\mca{W}_{1,\lambda}$ can be expressed as
\begin{align}
\mca{W}_{1,\lambda}(\vb*{x},\vb*{y})=\min\qty{\lambda(\|\vb*{x}-\tilde{\vb*{x}}\|_1+\|\vb*{y}-\tilde{\vb*{y}}\|_1)+\mca{W}_1(\tilde{\vb*{x}},\tilde{\vb*{y}})},
\end{align}
where the minimum is over all states $\tilde{\vb*{x}}$ and $\tilde{\vb*{y}}$ such that $\tilde{\vb*{x}}\le \vb*{x}$, $\tilde{\vb*{y}}\le \vb*{y}$, and $\|\tilde{\vb*{x}}\|_1=\|\tilde{\vb*{y}}\|_1$.
Here, $\vb*{x}\le \vb*{y}$ means that $x_n\le y_n$ for all $n$.
Consequently, computing $\mca{W}_{1,\lambda}$ is equivalent to solving the following minimization problem:
\begin{equation}\label{eq:lin.prob2}
\begin{array}{ll@{}}
\text{minimize}  & \tr{\msf{C}\pi}+\lambda\qty[\sum_n(x_n-\tilde{x}_n)+\sum_n(y_n-\tilde{y}_n)]\\
\text{subject to}& \pi\vb*{1}=\tilde{\vb*{y}},~\pi^\top\vb*{1}=\tilde{\vb*{x}}\\
                 & \tilde{\vb*{x}}\le\vb*{x},~\tilde{\vb*{y}}\le\vb*{y}\\
                 & \pi\ge 0.
\end{array}
\end{equation}
The linear programming problems in Eqs.~\eqref{eq:lin.prob1} and \eqref{eq:lin.prob2} can be efficiently solved using programming languages such as Python, Julia, or Mathematica.
A Mathematica code that computes the (generalized) Wasserstein distance can be found on GitHub \cite{Vu.2022.GH}.

\section{Speed limit for scalar observables}
Here we derive a speed limit for a scalar observable $\mca{O}_t$, defined in terms of state $\vb*{x}_t$ as
\begin{equation}
\mca{O}_t\coloneqq\sum_io_ix_i(t)=\vb*{o}^\top\vb*{x}_t,
\end{equation}
where $\vb*{o}=[o_1,\dots,o_N]^\top$ is a vector of real coefficients.
For convenience, we define the spectral norm and Lipschitz constant of vector $\vb*{o}$ as follows:
\begin{align}
\|\vb*{o}\|_\infty&\coloneqq\max_i|o_i|,\\
\|\vb*{o}\|_{\rm Lip}&\coloneqq\max_{\ev{i,j}\in\mca{E}}|o_i-o_j|.
\end{align}
We first consider the general case where external flows are present.
According to Prop.~\ref{prop:dual}, we obtain the following speed limit for observable $\mca{O}_t$ from Eq.~(\GenSL):
\begin{equation}\label{eq:obs.qsl}
\tau\ge\frac{|\mca{O}_\tau-\mca{O}_0|}{\max\qty{ \lambda^{-1}{\|\vb*{o}\|_\infty} , \|\vb*{o}\|_{\rm Lip} }\ev{\upsilon_{t,\lambda}}_\tau}.
\end{equation}
For the case that the external flows are absent, we can obtain a more compact speed limit by taking the $\lambda\to+\infty$ limit in Eq.~\eqref{eq:obs.qsl}, which reads
\begin{equation}
\tau\ge\frac{|\mca{O}_\tau-\mca{O}_0|}{\|\vb*{o}\|_{\rm Lip}\ev{\upsilon_{t}}_\tau}.
\end{equation}
This bound recovers the result reported in Ref.~\cite{Hamazaki.2022.PRXQ} [see Eq.~(19) therein].

\begin{proposition}\label{prop:dual}
For arbitrary states $\vb*{x}$, $\vb*{y}$ and a real vector $\vb*{o}$, the following inequality holds:
\begin{equation}\label{eq:prop.dual}
\qty| \vb*{o}^\top( \vb*{x} - \vb*{y} ) | \le \max\qty{ \lambda^{-1}{\|\vb*{o}\|_\infty} , \|\vb*{o}\|_{\rm Lip} }\mca{W}_{1,\lambda}(\vb*{x},\vb*{y}).
\end{equation}
\begin{proof}
Let $\tilde{\vb*{x}}$ and $\tilde{\vb*{y}}$ be two states that realize $\mca{W}_{1,\lambda}(\vb*{x},\vb*{y})$. In other words, we have
\begin{equation}
\mca{W}_{1,\lambda}(\vb*{x},\vb*{y})=\lambda( \|\vb*{x}-\tilde{\vb*{x}}\|_1 + \|\vb*{y}-\tilde{\vb*{y}}\|_1 ) + \mca{W}_1(\tilde{\vb*{x}},\tilde{\vb*{y}}).
\end{equation}
Then, by applying the triangle inequality and the Kantorovich--Rubinstein duality formula \cite{Villani.2008}
\begin{equation}
\mca{W}_1(\tilde{\vb*{x}},\tilde{\vb*{y}})=\max\qty{\vb*{\phi}^\top (\tilde{\vb*{x}} - \tilde{\vb*{y}}) ~|~\|\vb*{\phi}\|_{\rm Lip}\le 1},
\end{equation}
we can prove Eq.~\eqref{eq:prop.dual} as follows:
\begin{align}
&\qty| \vb*{o}^\top( \vb*{x} - \vb*{y} ) |\notag\\
&\le \qty| \vb*{o}^\top( \vb*{x} - \tilde{\vb*{x}} ) | + \qty| \vb*{o}^\top( \vb*{y} - \tilde{\vb*{y}} ) | + \qty| \vb*{o}^\top( \tilde{\vb*{x}} - \tilde{\vb*{y}} ) | \notag\\
&\le \|\vb*{o}\|_\infty(\|\vb*{x} - \tilde{\vb*{x}}\|_1+\|\vb*{y} - \tilde{\vb*{y}}\|_1)+\|\vb*{o}\|_{\rm Lip}\mca{W}_1(\tilde{\vb*{x}},\tilde{\vb*{y}})\notag\\
&\le \max\qty{ \lambda^{-1}{\|\vb*{o}\|_\infty} , \|\vb*{o}\|_{\rm Lip} }\mca{W}_{1,\lambda}(\vb*{x},\vb*{y}).
\end{align}
\end{proof}
\end{proposition}

\section{Chemical reaction networks}

\subsection{Graph construction for closed reaction networks}
Here we describe in detail the construction of an undirected graph $G(\mca{V},\mca{E})$ for the closed chemical system from which the Wasserstein distance can be defined immediately.
Notice the total mass conservation $\sum_im_i(\nu_i^{+\rho}-\nu_i^{-\rho})=0$ for any $\rho$.
The set of vertices is defined as $\mca{V}=\mca{S}$, where the vertex $i$ corresponds to the species $X_i$.
For each reaction channel $\rho$, define $\mca{S}^\rho_+\coloneqq\{i\,|\,\nu_i^{+\rho}>\nu_i^{-\rho}\}$ and $\mca{S}^\rho_-\coloneqq\{i\,|\,\nu_i^{+\rho}<\nu_i^{-\rho}\}$. Evidently, $\sum_{i\in\mca{S}^\rho_+}m_i(\nu_i^{+\rho}-\nu_i^{-\rho})=\sum_{i\in\mca{S}^\rho_-}m_i(\nu_i^{-\rho}-\nu_i^{+\rho})$.
According to Prop.~\ref{prop:zmn}, there exists a matrix $\msf{Z}^\rho=[z_{ij}^\rho]\in\mbb{R}^{|\mca{S}|\times|\mca{S}|}$ such that
\begin{align}
\sum_{j\in\mca{S}^\rho_-}z_{ij}^\rho&=m_i(\nu_i^{+\rho}-\nu_i^{-\rho}),~\forall i\in\mca{S}^\rho_+,\\
\sum_{i\in\mca{S}^\rho_+}z_{ij}^\rho&=m_j(\nu_j^{-\rho}-\nu_j^{+\rho}),~\forall j\in\mca{S}^\rho_-,\\
z_{ji}^\rho&=-z_{ij}^\rho\le 0,~\forall i\in\mca{S}^\rho_+,~j\in\mca{S}^\rho_-,\\
z_{ij}^\rho&=0,~\text{otherwise}.
\end{align}
For each $z_{ij}^\rho\neq 0$, we add an undirected edge $\ev{i,j}$ to $\mca{E}$ (see Fig.~\ref{fig:GraphIllust} for illustration).
Note that the existence of $\msf{Z}^\rho$ may not be unique; nevertheless, an instance of $\msf{Z}^\rho$ can always be found such that at most $|\mca{S}^\rho_+|+|\mca{S}^\rho_-|-1$ edges are added to the graph for each reaction channel $\rho$.
Repeating this process for all $\rho\in\mca{R}$, we can readily obtain the graph $G$.

\begin{figure}[t]
\centering
\includegraphics[width=0.95\linewidth]{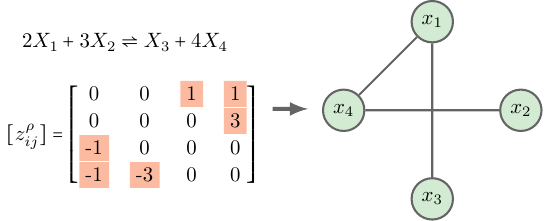}
\protect\caption{Illustration of the process of adding edges to the graph $G$ for a reaction channel $\rho$. Here, for simplicity, the molar mass is set to unity for all species (i.e., $m_i=1$ for any $i$).
For each $z_{ij}^\rho\neq 0$, an edge between $i$ and $j$ is added to $G$.}\label{fig:GraphIllust}
\end{figure}

By this construction, we can verify that
\begin{align}
	\dot{x}_i(t)&=\sum_{\rho}m_i(\nu^{+\rho}_{i}-\nu^{-\rho}_{i})J^{\rho}_t\notag\\
	&=\sum_{j\in\mca{B}_i}\sum_{\rho}z_{ij}^\rho J^{\rho}_t\notag\\
	&=\sum_{j\in\mca{B}_i}f_{ij}(t),
\end{align}
where $f_{ij}(t)=\sum_{\rho}z_{ij}^\rho J^{\rho}_t$.

\subsection{Derivation of Eq.~(\VeloUpB)}
Furthermore, by applying the triangle inequality, we can prove that the velocity is upper bounded as follows:
\begin{align}
\upsilon_t&=\sum_{\ev{i,j}\in\mca{E}}|f_{ij}(t)|\notag\\
&=\frac{1}{2}\sum_i\sum_{j\in\mca{B}_i}|f_{ij}(t)|\notag\\
&\le \frac{1}{2}\sum_i\sum_{j\in\mca{B}_i}\sum_\rho |z_{ij}^\rho||J_t^\rho|\notag\\
&=\frac{1}{2}\sum_\rho|J_t^\rho|\sum_i\sum_{j\in\mca{B}_i}|z_{ij}^\rho|\notag\\
&=\frac{1}{2}\sum_\rho|J_t^\rho|\sum_im_i|\nu^{+\rho}_{i}-\nu^{-\rho}_{i}|\notag\\
&=\sum_\rho\nu^\rho|J_t^\rho|,
\end{align}
where $\nu^\rho=(1/2)\sum_im_i|\nu^{+\rho}_{i}-\nu^{-\rho}_{i}|$.

\subsection{Generalization to generic reaction networks}
Here we derive the topological speed limit for generic (open) chemical reaction networks, in which the total mass concentration may not be conversed.
To this end, we construct the graph corresponding to the time evolution of $\vb*{x}_t$.
Once done, the topological speed limit Eq.~(\GenSL) is immediately obtained.
We need only consider reaction channels that do not conserve the total mass concentration because the conservative reactions have already been considered above.

Analogously, for each reaction channel $\rho$, we define $\mca{S}^\rho_+\coloneqq\{i\,|\,\nu_i^{+\rho}>\nu_i^{-\rho}\}$ and $\mca{S}^\rho_-\coloneqq\{i\,|\,\nu_i^{+\rho}<\nu_i^{-\rho}\}$. A non-conservative reaction channel implies $\sum_{i\in\mca{S}^\rho_+}m_i(\nu_i^{+\rho}-\nu_i^{-\rho})\neq\sum_{i\in\mca{S}^\rho_-}m_i(\nu_i^{-\rho}-\nu_i^{+\rho})$.
Then, according to Prop.~\ref{prop:zmn}, there exist a matrix $\msf{Z}^\rho=\{z_{ij}^\rho\}\in\mbb{R}^{|\mca{S}|\times|\mca{S}|}$ and a vector $\vb*{\mu}^\rho\in\mbb{R}^{|\mca{S}|}$ such that
\begin{align}
\sum_{j\in\mca{S}^\rho_-}z_{ij}^\rho&=\mu_i^{\rho}\le m_i(\nu_i^{+\rho}-\nu_i^{-\rho}),~\forall i\in\mca{S}^\rho_+,\\
\sum_{i\in\mca{S}^\rho_+}z_{ij}^\rho&=-\mu_j^{\rho}\le m_i(\nu_j^{-\rho}-\nu_j^{+\rho}),~\forall j\in\mca{S}^\rho_-,\\
z_{ji}^\rho&=-z_{ij}^\rho\le 0,~\forall i\in\mca{S}^\rho_+,~j\in\mca{S}^\rho_-,\\
z_{ij}^\rho&=0,~\text{otherwise},\\
\mu_i^\rho&\ge 0,~\forall i\in\mca{S}_+^\rho,\\
\mu_j^\rho&\le 0,~\forall j\in\mca{S}_-^\rho,\\
\mu_i^\rho&= 0,~\text{otherwise},\\
\sum_{i\in\mca{S}^\rho_+}\mu_i^{\rho}&=\min\qty{ \sum_{i\in\mca{S}^\rho_+}m_i(\nu_i^{+\rho}-\nu_i^{-\rho}), \, \sum_{i\in\mca{S}^\rho_-}m_i(\nu_i^{-\rho}-\nu_i^{+\rho}) }.
\end{align}
For each $z_{ij}^\rho\neq 0$, we add an undirected edge $\ev{i,j}$ to $\mca{E}$.
Again, the existence of $\msf{Z}^\rho$ and $\vb*{\mu}^{\rho}$ may not be unique; nevertheless, one can always find coefficients such that at most $|\mca{S}^\rho_+|+|\mca{S}^\rho_-|-1$ edges are added to the graph for each reaction channel $\rho$.
Repeating this process for all $\rho\in\mca{R}$, we readily obtain the graph $G$.

With the above construction, the time evolution of the reaction network can be expressed in the form of Eq.~(\DetEq) with
\begin{align}
f_i(t)&=\sum_\rho [m_i(\nu_i^{+\rho}-\nu_i^{-\rho})-\mu_i^{\rho}]J_t^\rho,\\
f_{ij}(t)&=\sum_\rho z_{ij}^\rho J_t^\rho.
\end{align}
For convenience, we define
\begin{align}
\chi^\rho&\coloneqq\min\qty{ \sum_{i\in\mca{S}^\rho_+}m_i(\nu_i^{+\rho}-\nu_i^{-\rho}), \, \sum_{i\in\mca{S}^\rho_-}m_i(\nu_i^{-\rho}-\nu_i^{+\rho}) },\\
\eta^\rho&\coloneqq\qty|\sum_{i\in\mca{S}^\rho_+}m_i(\nu_i^{+\rho}-\nu_i^{-\rho})-\sum_{i\in\mca{S}^\rho_-}m_i(\nu_i^{-\rho}-\nu_i^{+\rho})|.
\end{align}
Using these quantities, the velocity term $\upsilon_{t,\lambda}$ can be upper bounded as
\begin{align}
\upsilon_{t,\lambda}&=\lambda\sum_i|\sum_\rho [m_i(\nu_i^{+\rho}-\nu_i^{-\rho})-\mu_i^{\rho}]J_t^\rho|+\sum_{\ev{i,j}\in\mca{E}}|\sum_\rho z_{ij}^\rho J_t^\rho|\notag\\
&\le\sum_\rho|J_t^\rho|\qty(\lambda\sum_i|m_i(\nu_i^{+\rho}-\nu_i^{-\rho})-\mu_i^{\rho}|+\sum_{\ev{i,j}\in\mca{E}}|z_{ij}^\rho|)\notag\\
&=\sum_\rho(\lambda\eta^\rho + \chi^\rho)|J_t^\rho|.
\end{align}
Consequently, we obtain the following speed limit for generic reaction networks:
\begin{equation}\label{eq:gcrn.sl1}
\tau\ge\tau_1\coloneqq\frac{\mca{W}_{1,\lambda}(\vb*{x}_0,\vb*{x}_\tau)}{\ev{\sum_\rho(\lambda\eta^\rho + \chi^\rho)|J_t^\rho|}_\tau}.
\end{equation}
Define the generalized kinetic quantity $\ell_{t,\lambda}$ as
\begin{equation}
\ell_{t,\lambda}\coloneqq\sum_\rho(\lambda\eta^\rho + \chi^\rho)^2\frac{J_t^{-\rho}-J_t^{+\rho}}{\ln(J_t^{-\rho}/J_t^{+\rho})},
\end{equation}
we can prove that
\begin{equation}\label{eq:gcrn.tmp1}
\ev{\sum_\rho(\lambda\eta^\rho + \chi^\rho)|J_t^\rho|}_\tau\le\ev{\sqrt{\sigma_t\ell_{t,\lambda}}}_\tau\le\sqrt{\ev{\sigma_t}_\tau\ev{\ell_{t,\lambda}}_\tau}.
\end{equation}
Combining Eqs.~\eqref{eq:gcrn.sl1} and \eqref{eq:gcrn.tmp1} yields the following thermodynamic speed limit:
\begin{align}
\tau\ge\tau_2\coloneqq\frac{\mca{W}_{1,\lambda}(\vb*{x}_0,\vb*{x}_\tau)}{\sqrt{\ev{\sigma_t}_\tau\ev{\ell_{t,\lambda}}_\tau}},~\forall\lambda\ge 0.
\end{align}

So far, we have derived the speed limits for the mass concentration $\vb*{x}_t$. It is worth noting that speed limits for the molar concentration $\vb*{c}_t$ can be readily obtained by simply setting $m_i=1$ for all species $X_i$. Note that by this setting, the graph $G$ and coefficients $\chi^\rho$ and $\eta^\rho$ may change accordingly.
Eventually, the following inequalities hold:
\begin{equation}
\tau\ge\tau_1'\coloneqq\frac{\mca{W}_{1,\lambda}(\vb*{c}_0,\vb*{c}_\tau)}{\ev{\sum_\rho(\lambda\eta^\rho + \chi^\rho)|J_t^\rho|}_\tau}\ge \tau_2'\coloneqq\frac{\mca{W}_{1,\lambda}(\vb*{c}_0,\vb*{c}_\tau)}{\sqrt{\ev{\sigma_t}_\tau\ev{\ell_{t,\lambda}}_\tau}}.
\end{equation}
Next, we show that the $\lambda=1/2$ speed limit is always tighter than a thermodynamic bound reported in Ref.~\cite{Yoshimura.2021.PRL}, which reads
\begin{equation}
\tau\ge\tau_3\coloneqq\frac{\mca{T}(\vb*{c}_0,\vb*{c}_\tau)}{\sqrt{\ev{\sigma_t}_\tau\ev{d_t}_\tau}},
\end{equation}
where $\mca{T}(\vb*{c}_0,\vb*{c}_\tau)\coloneqq\|\vb*{c}_0-\vb*{c}_\tau\|_1/2$ is the total variation distance and $d_t\coloneqq(|\mca{S}|/8)\sum_{\rho,i}(\nu_i^{+\rho}-\nu_i^{-\rho})^2(J_t^{+\rho}+J_t^{-\rho})$ corresponds to the diffusion coefficient.
Note that $(1/2)\eta^\rho + \chi^\rho=(1/2)\sum_i|\nu_i^{+\rho}-\nu_i^{-\rho}|$.
Applying the Cauchy--Schwarz inequality and the inequality $(x-y)/\ln(x/y)\le (x+y)/2$, we have
\begin{align}
\ell_{t,1/2}&=\frac{1}{4}\sum_{\rho}\qty(\sum_i|\nu_i^{+\rho}-\nu_i^{-\rho}|)^2\frac{J^{-\rho}_t-J^{+\rho}_t}{\ln(J^{-\rho}_t/J^{+\rho}_t)}\notag\\
&\le\frac{|\mca{S}|}{8}\sum_{\rho,i}(\nu_i^{+\rho}-\nu_i^{-\rho})^2(J^{-\rho}_t+J^{+\rho}_t)\notag\\
&=d_t.
\end{align}
By combining the above inequality and the relation $\mca{W}_{1,1/2}(\vb*{c}_0,\vb*{c}_\tau)=\mca{T}(\vb*{c}_0,\vb*{c}_\tau)$, we readily obtain the following hierarchical relationship:
\begin{equation}
\tau_1'\ge\tau_2'\ge\tau_3.
\end{equation}

\begin{proposition}\label{prop:zmn}
Let $\vb*{a}=[a_1,\dots,a_n]^\top$ and $\vb*{b}=[b_1,\dots,b_m]^\top$ be two vectors of positive numbers.
Then there exist matrix $\msf{Z}=[z_{ij}]\in\mbb{R}_{\ge 0}^{n\times m}$ and nonnegative vectors $\tilde{\vb*{a}}=[\tilde{a}_1,\dots,\tilde{a}_n]^\top$ and $\tilde{\vb*{b}}=[\tilde{b}_1,\dots,\tilde{b}_m]^\top$ such that the following conditions are satisfied:
\begin{align}
\sum_{j=1}^mz_{ij}&=\tilde{a}_i\le a_i,~\forall 1\le i\le n,\label{eq:prop.zmn.tmp1}\\
\sum_{i=1}^nz_{ij}&=\tilde{b}_j\le b_j,~\forall 1\le j\le m,\label{eq:prop.zmn.tmp2}\\
\sum_i\tilde{a}_i&=\min\qty{ \sum_ia_i,\sum_jb_j }.\label{eq:prop.zmn.tmp3}
\end{align}
\end{proposition}
\begin{proof}
Without loss of generality, we can assume that $a_1\le\dots\le a_n$ and $b_1\le\dots\le b_m$.
We prove by induction on $k=m+n\ge 2$.
In the case of $k=2$ (i.e., $m=n=1$), we can set $z_{11}=\tilde{a}_1=\tilde{b}_1=\min\{a_1,b_1\}$.
Suppose that the statement holds for all $k\le\bar{k}$.
We consider an arbitrary case with $k=\bar{k}+1$.
Assume that $a_1\le b_1$. If $n=1$ then we can choose $z_{11}=a_1$, $z_{1j}=0$ for all $j>1$, $\tilde{\vb*{a}}=[a_1]$, and $\tilde{\vb*{b}}=[a_1,0,\dots,0]^\top$.
If $n\ge 2$, then consider two vectors $\vb*{a}'=[a_2,\dots,a_n]^\top$ and $\vb*{b}'=[b_1-a_1,b_2,\dots,b_m]^\top$. There exist $\msf{Z}'$, $\tilde{\vb*{a}}'$, and $\tilde{\vb*{b}}'$ such that
\begin{align}
\sum_{j=1}^mz_{ij}'&=\tilde{a}_i'\le a_i',~\forall 1\le i\le n-1,\\
\sum_{i=1}^{n-1}z_{ij}'&=\tilde{b}_j'\le b_j',~\forall 1\le j\le m,\\
\sum_i\tilde{a}_i'&=\min\qty{ \sum_ia_i',\sum_jb_j' }.
\end{align}
We construct $z_{11}=a_1$, $z_{1j}=0$ for all $j>1$, $z_{ij}=z_{(i-1)j}'$ for all $i\ge 2$, $\tilde{\vb*{a}}=[a_1,\tilde{a}_1',\dots,\tilde{a}_{n-1}']^\top$, and $\tilde{\vb*{b}}=[a_1+\tilde{b}_1',\tilde{b}_2',\dots,\tilde{b}_{m}']^\top$.
It is easy to verify that this combination satisfies all conditions \eqref{eq:prop.zmn.tmp1}, \eqref{eq:prop.zmn.tmp2}, and \eqref{eq:prop.zmn.tmp3}.
\end{proof}
For example, for two vectors $\vb*{a}=[4,5]^\top$ and $\vb*{b}=[1,2,3]^\top$, we can construct an instance of matrix $\msf{Z}$ and vectors $\tilde{\vb*{a}},\tilde{\vb*{b}}$ as follows:
\begin{align}
\msf{Z}&=\begin{bmatrix}
	1 & 2 & 1\\
	0 & 0 & 2
\end{bmatrix},~\tilde{\vb*{a}}=[4,2]^\top,~\tilde{\vb*{b}}=[1,2,3]^\top.
\end{align}

\section{Bosonic transport}
\subsection{Time evolution of boson number}
Using the relation $[b_i,\hat n_i]=b_i$ and $[\hat n_i,h_Z]=0$, we can calculate as follows:
\begin{align}
	-i\tr{\hat n_i [H_t,\varrho_t]}&=-i\tr{[\hat n_i,H_t]\varrho_t}\notag\\
	&=i\gamma\sum_{j\in\mca{B}_i}\tr{[\hat n_i,b_i^\dagger b_j+b_j^\dagger b_i]\varrho_t}\notag\\
	&=i\gamma\sum_{j\in\mca{B}_i}\tr{(b_i^\dagger b_j-b_j^\dagger b_i)\varrho_t}\notag\\
	&=2\gamma\sum_{j\in\mca{B}_i}\Im\qty[\tr{b_j^\dagger b_i\varrho_t}].
\end{align}
Similarly, by noting that $[b_i,b_i^\dagger]=1$, we have
\begin{align}
&\tr{\hat n_i\sum_{j\in\Lambda}\qty(\mca{D}[L_{j,+}]+\mca{D}[L_{j,-}])\varrho_t}\notag\\
&=\tr{\hat n_i\sum_{j\in\Lambda}\gamma_{j,+}\qty(b_j^\dagger\varrho_tb_j-\{b_jb_j^\dagger,\varrho_t\}/2)}\notag\\
&+\tr{\hat n_i\sum_{j\in\Lambda}\gamma_{j,-}\qty(b_j\varrho_tb_j^\dagger-\{b_j^\dagger b_j,\varrho_t\}/2)}\notag\\
&=\tr{\hat n_i\gamma_{i,+}\qty(b_i^\dagger\varrho_tb_i-\{\hat n_i+1,\varrho_t\}/2)}\notag\\
&+\tr{\hat n_i\gamma_{i,-}\qty(b_i\varrho_tb_i^\dagger-\{\hat n_i,\varrho_t\}/2)}\notag\\
&=\gamma_{i,+}\tr{b_i^\dagger\varrho_t(b_i+\hat n_ib_i)-\{\hat n_i^2+\hat n_i,\varrho_t\}/2}\notag\\
&+\gamma_{i,-}\tr{(b_i\hat n_i - b_i)\varrho_tb_i^\dagger-\{\hat n_i^2,\varrho_t\}/2}\notag\\
&=\gamma_{i,+}\tr{b_i^\dagger\varrho_tb_i}-\gamma_{i,-}\tr{b_i\varrho_tb_i^\dagger}\notag\\
&=\tr{L_{i,+}\varrho_tL_{i,+}^\dagger}-\tr{L_{i,-}\varrho_tL_{i,-}^\dagger}.
\end{align}
Taking the time derivative of $x_i(t)=\tr{\hat n_i \varrho_t}$, we can calculate the time evolution of $x_i(t)$ as follows:
\begin{align}
\dot{x}_i(t)&=\tr{\hat n_i \dot{\varrho}_t}\notag\\
	&=\tr{-i\hat n_i [H_t,\varrho_t]+\hat n_i\sum_{i\in\Lambda}\qty(\mca{D}[L_{i,+}]+\mca{D}[L_{i,-}])\varrho_t}\notag\\
	&=2\gamma\sum_{j\in\mca{B}_i}\Im\qty[\tr{b_j^\dagger b_i\varrho_t}]\notag\\
	&+\tr{L_{i,+}\varrho_tL_{i,+}^\dagger}-\tr{L_{i,-}\varrho_tL_{i,-}^\dagger}\notag\\
	&=f_i(t)+\sum_{j\in\mca{B}_i}f_{ij}(t),
\end{align}
where 
\begin{align}
	f_i(t)&=\tr{L_{i,+}\varrho_tL_{i,+}^\dagger}-\tr{L_{i,-}\varrho_tL_{i,-}^\dagger},\\
	f_{ij}(t)&=2\gamma\Im\qty[\tr{b_j^\dagger b_i\varrho_t}].
\end{align}
It can be verified that $f_{ij}(t)=-f_{ji}(t)$.

\subsection{Derivation of Eq.~(\BosonVtUp)}
Applying the Cauchy--Schwarz inequality, we can upper bound $|f_{ij}(t)|$ as follows:
\begin{align}
|f_{ij}(t)|&\le 2\gamma\sqrt{\tr{b_i^\dagger b_i\varrho_t}\tr{b_j^\dagger b_j\varrho_t}}\notag\\
&\le \gamma\qty( \tr{b_i^\dagger b_i\varrho_t} + \tr{b_j^\dagger b_j\varrho_t} )\notag\\
&=\gamma[x_i(t)+x_j(t)].
\end{align}
Taking the sum over all edges $\ev{i,j}\in\mca{E}$, we obtain
\begin{align}\label{eq:vt.ub.boson.tmp1}
\sum_{\ev{i,j}}|f_{ij}(t)|&\le \sum_{\ev{i,j}}\gamma[x_i(t)+x_j(t)]\notag\\
&\le\gamma d_G\sum_ix_i(t)=\gamma d_G\mca{N}_t.
\end{align}
In addition, according to Prop.~\ref{prop:current.bound}, we have
\begin{align}\label{eq:vt.ub.boson.tmp2}
\sum_i|f_i(t)|&=\sum_i\qty|\tr{L_{i,+}\varrho_tL_{i,+}^\dagger}-\tr{L_{i,-}\varrho_tL_{i,-}^\dagger}|\notag\\
&\le\frac{\sigma_t}{2}\Phi\qty(\frac{\sigma_t}{2a_t})^{-1}.
\end{align}
By combining the inequalities \eqref{eq:vt.ub.boson.tmp1} and \eqref{eq:vt.ub.boson.tmp2}, the velocity $\upsilon_{t,\lambda}$ can be upper bounded as
\begin{equation}
\upsilon_{t,\lambda}\le \gamma d_G\mca{N}_t+\lambda\frac{\sigma_t}{2}\Phi\qty(\frac{\sigma_t}{2a_t})^{-1}.
\end{equation}

\begin{figure*}[t]
\centering
\includegraphics[width=1\linewidth]{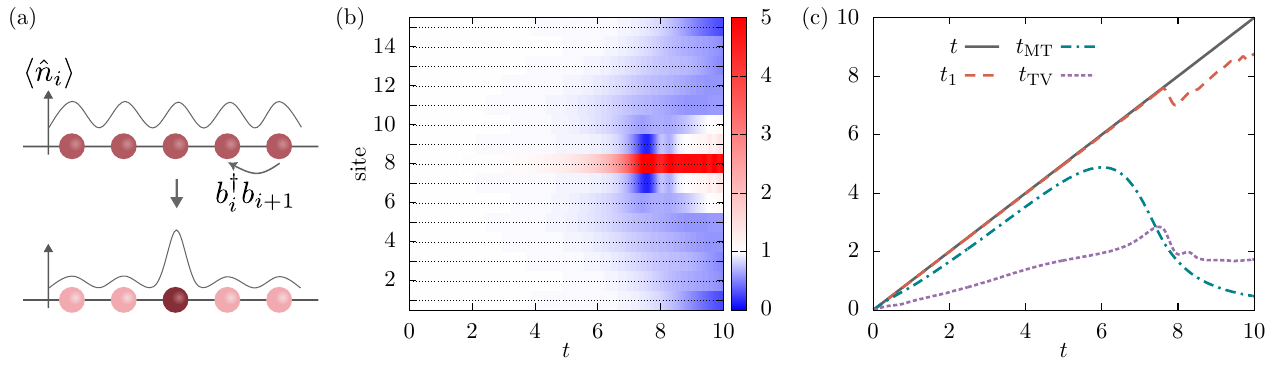}
\protect\caption{(a) Schematic of the one-dimensional Bose--Hubbard chain with a finite length. Bosons can hop from site $i$ to $i + 1$ and vice versa. (b) Time variation of the boson number at each site $i\in[1,15]$. The bosonic system is initiated in the ground state of $H_0$, in which bosons are almost uniformly distributed. After a period of time $\tau$, bosons tend to gather at the central site. (c) Numerical verification of the speed limits. The operational time $t$, topological bound $t_1$, and non-topological bounds $t_{\rm MT}$ and $t_{\rm TV}$ are depicted by solid, dashed, and dash-dotted and dotted lines, respectively. The parameters are fixed at $U_0=10$, $\gamma=1$, and $U_\tau=[10,9.5,8.5,7, 5, 2, 0.5, -5, 0.5, 2, 5, 7, 8.5, 9.5, 10]^\top$.}\label{fig:BosonChain}
\end{figure*}

\begin{proposition}\label{prop:current.bound}
The following inequality holds:
\begin{equation}
\sum_i\qty|\tr{L_{i,+}\varrho_tL_{i,+}^\dagger}-\tr{L_{i,-}\varrho_tL_{i,-}^\dagger}|\le{\frac{\sigma_t}{2}\Phi\qty(\frac{\sigma_t}{2a_t})^{-1}},
\end{equation}
where $\Phi(x)$ is the inverse function of $x\tanh(x)$.
\end{proposition}
\begin{proof}
Let $\varrho_t=\sum_np_n(t)\dyad{n_t}$ be the spectral decomposition of the density operator $\varrho_t$.
Define $r_{mn}^{i,\pm}(t)\coloneqq|\mel{m_t}{L_{i,\pm}}{n_t}|^2$, then the rates of entropy production and dynamical activity can be expressed as
\begin{align}
\sigma_t&=\sum_i\sum_{m,n}[r_{mn}^{i,+}(t)p_n(t)-r_{nm}^{i,-}(t)p_m(t)]\ln\frac{r_{mn}^{i,+}(t)p_n(t)}{r_{nm}^{i,-}(t)p_m(t)}\notag\\
&\eqqcolon\sum_i\sum_{m,n}\sigma_{mn}^{i}(t),\\
a_t&=\sum_{i}\sum_{m,n}[r_{mn}^{i,+}(t)p_n(t)+r_{nm}^{i,-}(t)p_m(t)]\notag\\
&\eqqcolon\sum_i\sum_{m,n}a_{mn}^{i}(t).
\end{align}
Note that $x\Phi(x/y)^{-1}$ is a is a concave function over $(0,+\infty)\times(0,+\infty)$. Applying Jensen’s inequality yields
\begin{align}
&\sum_i\qty|\tr{L_{i,+}\varrho_tL_{i,+}^\dagger}-\tr{L_{i,-}\varrho_tL_{i,-}^\dagger}|\notag\\
&=\sum_i\qty|\sum_{m,n}[r_{mn}^{i,+}(t)p_n(t)-r_{nm}^{i,-}(t)p_m(t)] |\notag\\
&\le \sum_{i}\sum_{m,n}\qty|r_{mn}^{i,+}(t)p_n(t)-r_{nm}^{i,-}(t)p_m(t)|\notag\\
&=\sum_{i}\sum_{m,n}\frac{\sigma_{mn}^i(t)}{2}\Phi\qty(\frac{\sigma_{mn}^i(t)}{2a_{mn}^i(t)})^{-1}\notag\\
&\le \frac{\sigma_t}{2}\Phi\qty(\frac{\sigma_t}{2a_t})^{-1},
\end{align}
which completes the proof.
\end{proof}

\subsection{Numerical demonstration}
Here we numerically demonstrate the speed limits for bosonic transport in a one-dimensional Bose--Hubbard model with the size of $N=15$ [see Fig.~\ref{fig:BosonChain}(a) for illustration].
Bosons can hop from site $i$ to site $j$ if and only if $|i-j|=1$.
The Hamiltonian is given by
\begin{equation}
H_t=-\gamma\sum_{i=1}^{N-1}(b_i^\dagger b_{i+1}+b_{i+1}^\dagger b_i)+\sum_{i=1}^{N}U_i(t)\frac{\hat n_i(\hat n_i-1)}{2}.
\end{equation}
For this isolated bosonic system, the flows are determined as follows:
\begin{align}
	f_i(t)&=0,\\
	f_{ij}(t)&=\delta_{|i-j|,1}2\gamma\Im\qty[\tr{b_j^\dagger b_i\varrho_t}],
\end{align}
and the maximal vertex degree is $d_G=2$.
Here $\delta_{x,y}$ denotes the Kronecker delta.
The velocity $\upsilon_t$ is thus given by
\begin{equation}
\upsilon_t=2\gamma\sum_{i=1}^{N-1}\qty|\Im\qty[\tr{b_{i+1}^\dagger b_i\varrho_t}]|.
\end{equation}
Consequently, the derived speed limit reads
\begin{align}
\tau&\ge\frac{\mca{W}_1({\vb*{x}}_0,{\vb*{x}}_\tau)}{\ev{\upsilon_t}_\tau}\eqqcolon\tau_1.
\end{align}
For comparison, we also examine two other non-topological speed limits. The first one is the generalization of the Mandelstam--Tamm bound \cite{Uhlmann.1992.PLA}, which reads
\begin{equation}
\tau\ge\frac{\mca{L}(\varrho_0,\varrho_\tau)}{\ev{\Delta_{{\varrho}_t} H_t}_\tau}\eqqcolon\tau_{\rm MT},
\end{equation}
where $\mca{L}(\varrho,\sigma)\coloneqq\arccos\tr{\qty|{\varrho^{1/2}\sigma^{1/2}}|}$ is the Bures angle.
The second one is a speed limit using the total variation distance $\mca{T}({\vb*{x}},{\vb*{y}})\coloneqq\|\vb*{x}-\vb*{y}\|_1/2$, which does not consider the topological nature and is given by
\begin{equation}
\tau\ge\frac{\mca{T}({\vb*{x}}_0,{\vb*{x}}_\tau)}{\ev{\upsilon_t}_\tau}\eqqcolon\tau_{\rm TV}.
\end{equation}
Notice that $\tau_1\ge\tau_{\rm TV}$ since $\mca{W}_1({\vb*{x}}_0,{\vb*{x}}_\tau)\ge\mca{T}({\vb*{x}}_0,{\vb*{x}}_\tau)$.

We set the initial state to the ground state of the Hamiltonian $H_0$, which can be calculated using the algorithm of density matrix renormalization group \cite{Fishman.2022.SPPC}. Initially, bosons are almost uniformly distributed.
The time-dependent Hamiltonian is modulated as $H_t=(1-t/\tau)H_0+(t/\tau)H_\tau$, where $H_\tau$ is the final Hamiltonian.
More specifically, the interaction coefficients $U_i(t)~(1\le i\le N)$ are given by
\begin{equation}
U_i(t)=\qty(1-\frac{t}{\tau})U_0+\frac{t}{\tau}U_{\tau,i},
\end{equation}
where $\{U_{\tau,i}\}$ are the interaction coefficients of the final Hamiltonian $H_\tau$.
By this protocol, bosons tend to gather at the central site after a period of time $\tau$ [see Fig.~\ref{fig:BosonChain}(b)].
We employ the time dependent variational principle algorithm \cite{Yang.2020.PRB} to simulate the time evolution of the bosonic system.
For each time $t\le\tau\,(=10)$, we calculate all relevant quantities and plot the bounds in Fig.~\ref{fig:BosonChain}(c).
As shown, the topological bound $t\ge t_1$ is tight and can be saturated.
In contrast, the non-topological bound $t_{\rm TV}$ cannot accurately capture the operational time of bosonic transport. This is because the total variation distance does not consider topological information as the Wasserstein distance does. Although the Mandelstam--Tamm bound $t\ge t_{\rm MT}$ yields good prediction for $t\le 6$, it subsequently becomes trivial as $t$ increases.

\section{Isolated quantum systems}
Here we derive a topological speed limit for isolated systems whose dynamics is described by the von Neumann equation:
\begin{equation}
\dot\varrho_t=-i[H_t,\varrho_t],
\end{equation}
where $H_t$ is a time-dependent Hamiltonian.
We consider a complete orthogonal set of projection operators $\{P_n\}_n$, i.e., $P_mP_n=\delta_{mn}P_n$ and $\sum_nP_n=\mbb{1}$.
Since the dynamics is invariant under transformation $H_t\to H_t+\alpha_t\mbb{1}$, we can assume that $\tr{H_t\tilde{\varrho}_t}=0$, where $\tilde{\varrho}_t\coloneqq\sum_nP_n\varrho_tP_n$ is a density matrix projected on the space of $\{P_n\}_n$.
Let $G(\mca{V},\mca{E})$ be the graph, in which $\mca{V}=\{n\}_n$ and $\mca{E}=\{\ev{m,n}\,|\,P_mH_tP_n\neq 0,\,m<n\}$, and $d_G$ be the maximum vertex degree of the graph.
Considering a vector of projective measurements, $x_n(t)=\tr{P_n\varrho_t}$, we obtain the time evolution of $x_n(t)$ in the following equation:
\begin{align}
\dot x_n(t)&=-i\sum_{m(\neq n)}\qty(\tr{P_nH_tP_m\varrho_t}-\tr{P_mH_tP_n\varrho_t})\notag\\
&=\sum_{m\in\mca{B}_n}f_{nm}(t),
\end{align}
where $f_{mn}(t)=-f_{nm}(t)$ is given by
\begin{equation}
f_{nm}(t)=-i\qty(\tr{P_nH_tP_m\varrho_t}-\tr{P_mH_tP_n\varrho_t}).
\end{equation}
Applying the speed limit in Eq.~(\ConSL), we obtain the following bound on the operational time:
\begin{equation}\label{eq:iso.qsl1}
\tau\ge\frac{\mca{W}_1(\vb*{x}_0,\vb*{x}_\tau)}{\ev{\upsilon_{t}}_\tau},
\end{equation}
where $\upsilon_{t}$ can be explicitly expressed as
\begin{equation}
\upsilon_{t}=\sum_{\ev{m,n}\in\mca{E}}\qty|\tr{(P_nH_tP_m-P_mH_tP_n)\varrho_t}|.
\end{equation}

By further upper bounding $\upsilon_{t}$, we can also obtain another speed limit that resembles the Mandelstam--Tamm bound.
By applying the Cauchy--Schwarz inequality, we have
\begin{align}
|\tr{P_nH_tP_m\varrho_t}|&\le\tr{P_mH_tP_nH_tP_m\varrho_t}^{1/2}\tr{P_n\varrho_t}^{1/2},\\
|\tr{P_mH_tP_n\varrho_t}|&\le\tr{P_nH_tP_mH_tP_n\varrho_t}^{1/2}\tr{P_m\varrho_t}^{1/2}.
\end{align}
Combining the above inequalities and applying the Cauchy--Schwarz inequality yield
\begin{align}
|f_{nm}(t)|&\le\qty[\tr{(P_mH_tP_nH_tP_m+P_nH_tP_mH_tP_n)\varrho_t}]^{1/2}\notag\\
&\times\qty[\tr{(P_n+P_m)\varrho_t}]^{1/2}.
\end{align}
Therefore, using the fact $\sum_nP_n=\mbb{1}$, we have
\begin{align}
\upsilon_{t}&\le \sum_{\ev{m,n}\in\mca{E}}\qty[\tr{(P_mH_tP_nH_tP_m+P_nH_tP_mH_tP_n)\varrho_t}]^{1/2}\notag\\
&\times\qty[\tr{(P_n+P_m)\varrho_t}]^{1/2}\notag\\
&\le \qty[\sum_{\ev{m,n}\in\mca{E}}\tr{(P_mH_tP_nH_tP_m+P_nH_tP_mH_tP_n)\varrho_t}]^{1/2}\notag\\
&\times\qty[\sum_{\ev{m,n}\in\mca{E}}\tr{(P_n+P_m)\varrho_t}]^{1/2}\notag\\
&\le \qty[\sum_{n}\tr{H_t^2P_n\varrho_tP_n}]^{1/2}\sqrt{d_G}\notag\\
&=\sqrt{d_G}\Delta_{\tilde{\varrho}_t} H_t,
\end{align}
where ${\Delta_\varrho H_t}\coloneqq\sqrt{\tr{H_t^2\varrho}-(\tr{H_t\varrho})^2}$ is the energy fluctuation of the Hamiltonian $H_t$ with respect to $\varrho$.
Consequently, we obtain the following speed limit:
\begin{equation}\label{eq:iso.qsl2}
\tau\ge\frac{\mca{W}_1(\vb*{x}_0,\vb*{x}_\tau)}{\sqrt{d_G}\ev{\Delta_{\tilde{\varrho}_t} H_t}_\tau}.
\end{equation}
We note that this speed limit is similar but different from the generalization of the Mandelstam--Tamm bound \cite{Uhlmann.1992.PLA}, which reads
\begin{equation}\label{eq:MT.qsl}
\tau\ge\frac{\mca{L}(\varrho_0,\varrho_\tau)}{\ev{\Delta_{{\varrho}_t} H_t}_\tau},
\end{equation}
where $\mca{L}(\varrho,\sigma)\coloneqq\arccos\tr{\qty|{\varrho^{1/2}\sigma^{1/2}}|}$ is the Bures angle.
Notably, the new speed limit can be applied to quantum systems under measurement, making it more applicable than the generalized Mandelstam--Tamm bound.

We now consider a specific case where $P_n=\dyad{n}$. In this case, $|f_{mn}(t)|$ can be upper bounded as
\begin{align}
|f_{mn}(t)|&=|(\mel{n}{H_t}{m}\mel{m}{\varrho_t}{n}-\mel{m}{H_t}{n}\mel{n}{\varrho_t}{m})|\notag\\
&\le 2|\mel{m}{H_t}{n}||\mel{n}{\varrho_t}{m}|.
\end{align}
According to the Cauchy--Schwarz inequality, we have
\begin{equation}
|\mel{n}{\varrho_t}{m}|=|\mel{n}{\varrho_t^{1/2}\varrho_t^{1/2}}{m}|\le\sqrt{\mel{n}{\varrho_t}{n}\mel{m}{\varrho_t}{m}}.
\end{equation}
Therefore, 
\begin{align}
|f_{mn}(t)|\le 2|\mel{m}{H_t}{n}|\sqrt{\mel{n}{\varrho_t}{n}\mel{m}{\varrho_t}{m}}.
\end{align}
Consequently, the velocity can be upper bounded in terms of the Hamiltonian and diagonal terms of density matrix as
\begin{align}
\upsilon_t\le 2\sum_{\ev{m,n}\in\mca{E}}|h_{mn}(t)|\sqrt{p_n(t)p_m(t)}\eqqcolon \check{\upsilon}_t,
\end{align}
where we have defined $h_{mn}(t)\coloneqq \mel{m}{H_t}{n}$ and $p_n(t)\coloneqq\mel{n}{\varrho_t}{n}$ for simplicity.
Then, Eq.~\eqref{eq:iso.qsl1} yields the following speed limit:
\begin{align}
\tau\ge\frac{\mca{W}_1(\vb*{x}_0,\vb*{x}_\tau)}{\ev{\check{\upsilon}_t}_\tau},\label{eq:iso.qsl3}
\end{align}
which is easier to compute than the original bound \eqref{eq:iso.qsl1}.

\section{Measurement-induced quantum walk}
Here we demonstrate an application of the speed limits \eqref{eq:iso.qsl1}, \eqref{eq:iso.qsl2}, and \eqref{eq:iso.qsl3} for a quantum system under measurement.
We consider a model of the continuous-time quantum walk \cite{Didi.2022.PRE}, which is induced by measurements performed at discrete times.
The system Hamiltonian is given by
\begin{equation}
H=\sum_{n=1}^{N-1}\gamma_n(\dyad{n}{n+1}+\dyad{n+1}{n}),
\end{equation}
which describes hops between nearest neighbors on a finite line lattice.
The projective measurements with operators $\{\dyad{n}\}_n$ are performed at times $t_k\coloneqq k\Delta t$ for $k=0,\dots,K\,(=\tau/\Delta t)$. Between each measurement event, the system unitarily evolves according to the von Neumman equation.
The density matrix after the $k$th measurement is given by
\begin{equation}
\varrho_{t_k^+}=\sum_n\mel{n}{\varrho_{t_k}}{n}\dyad{n}.
\end{equation}
For this system, we consider the set of projection operators $\{P_n=\dyad{n}\}$. The graph $G$ thus has $N$ vertices and $N-1$ edges that connect $n$ and $n+1$ for all $1\le n\le N-1$. The maximum degree of the graph is $d_G=2$.
Applying Eqs.~\eqref{eq:iso.qsl1} and \eqref{eq:iso.qsl2} to the time evolution of the system between the $k$th and $(k+1)$th measurements yields the following result:
\begin{equation}\label{eq:mea.ind.tmp1}
\sqrt{2}\int_{t_k}^{t_{k+1}}\Delta_{\tilde{\varrho}_t}H\dd{t}\ge \int_{t_k}^{t_{k+1}}\upsilon_{t}\dd{t}\ge\mca{W}_1(\vb*{x}_{t_k},\vb*{x}_{t_{k+1}}).
\end{equation}
Taking the sum of both sides of Eq.~\eqref{eq:mea.ind.tmp1} for $k=0,\dots,K-1$ and applying the triangle inequality for $\mca{W}_1$, we obtain
\begin{equation}
\sqrt{2}\int_{0}^{\tau}\Delta_{\tilde{\varrho}_t}H\dd{t}\ge \int_{0}^{\tau}\upsilon_{t}\dd{t}\ge\mca{W}_1(\vb*{x}_{0},\vb*{x}_{\tau}).
\end{equation}
Consequently, the following speed limits are derived:
\begin{equation}
\tau\ge\frac{\mca{W}_1(\vb*{x}_0,\vb*{x}_\tau)}{\ev{\upsilon_{t}}_\tau}\ge\frac{\mca{W}_1(\vb*{x}_0,\vb*{x}_\tau)}{\sqrt{2}\ev{\Delta_{\tilde{\varrho}_t} H}_\tau}.
\end{equation}
We note that the system is measured at discrete times; therefore, the generalized Mandelstam--Tamm bound cannot be applied directly.

By applying the speed limit \eqref{eq:iso.qsl3}, we can also obtain another meaningful bound.
Note that $\mca{E}=\{\ev{n,n+1}\,|\,1\le n\le N-1\}$ and $h_{mn}(t)=\gamma_{\min\{m,n\}}\delta_{|m-n|,1}$ in this system.
The quantity $\check{\upsilon}_t$ can be upper bounded as follows:
\begin{align}
\check{\upsilon}_t&=2\sum_{n=1}^{N-1}\gamma_n\sqrt{p_n(t)p_{n+1}(t)}\notag\\
&\le \sum_{n=1}^{N-1}\gamma_n[p_n(t)+p_{n+1}(t)]\notag\\
&\le 2\max_{n}\gamma_n.
\end{align}
Therefore, the operational time is lower bounded as
\begin{equation}
\tau\ge\frac{\mca{W}_1(\vb*{x}_0,\vb*{x}_\tau)}{2\max_{n}\gamma_n}.
\end{equation}

\section{Quantum communication using spin systems}
Here we derive a topological speed limit for quantum communication through an arbitrary graph $G(\mca{V},\mca{E})$ of spins with ferromagnetic Heisenberg interactions \cite{Bose.2003.PRL}.
The vertices $\{n\in\mca{V}\}_n$ represent spins and the edges $\{\ev{n,m}\in\mca{E}\}$ connect interacting spins.
The Hamiltonian is given by
\begin{equation}
H_t=-\frac{\gamma}{2}\sum_{\ev{n,m}\in\mca{E}}\vec{\sigma}_n\cdot\vec{\sigma}_m + \sum_{n\in\mca{V}}B_n(t)\sigma_n^z,
\end{equation}
where $\gamma>0$ denotes the coupling strength, $\vec{\sigma}_n=[\sigma_n^x,\sigma_n^y,\sigma_n^z]^\top$ is the vector of Pauli spin operators for the $n$th spin, and $B_n(t)$ is an external magnetic field in the $z$ direction.

We assume that the sender, Alice, has a quantum state encoded in spin $1$ and wants to relay it to the receiver, Bob, who can access and read out spin $N$. By manipulating the external magnetic field (which Alice can control), quantum information can be transmitted through the graph of spins. After a predetermined time when the state of spin $1$ is transferred to spin $N$, Bob reads out the state of this site. The minimum time required for high-fidelity information transmission is the quantity of interest.

Consider state $\vb*{x}_t$ with $x_n(t)=(\tr{\sigma_n^z\varrho_t}+1)/2\ge 0$. Notice that $x_n(t)=0$ ($x_n(t)=1$) corresponds to spin down (up) with respect to the $z$ direction.
Taking the time derivative of $x_i(t)$, we can calculate as follows:
\begin{align}
\dot{x}_n(t)&=\tr{\sigma_n^z\dot{\varrho}_t}/2\notag\\
&=(i/2)\tr{[\sigma_n^z,H_t]\varrho_t}\notag\\
&=-(i\gamma/4)\sum_{m\in\mca{B}_n}\tr{[\sigma_n^z,\vec{\sigma}_n\cdot\vec{\sigma}_m]\varrho_t}\notag\\
&=(\gamma/2)\sum_{m\in\mca{B}_n}\tr{(\sigma_n^y\sigma_m^x-\sigma_n^x\sigma_m^y)\varrho_t}\notag\\
&=\sum_{m\in\mca{B}_n}f_{nm}(t),
\end{align}
where $f_{nm}(t)=-f_{mn}(t)$ is given by
\begin{equation}
f_{nm}(t)=(\gamma/2)\tr{(\sigma_n^y\sigma_m^x-\sigma_n^x\sigma_m^y)\varrho_t}.
\end{equation}
Therefore, the velocity can be expressed as
\begin{equation}
\upsilon_t=\frac{\gamma}{2}\sum_{\ev{m,n}\in\mca{E}}\qty|\tr{(\sigma_n^y\sigma_m^x-\sigma_n^x\sigma_m^y)\varrho_t}|.
\end{equation}
Consequently, the derived topological speed limit reads
\begin{equation}\label{eq:spin.topo.sl}
\tau\ge\frac{\mca{W}_1(\vb*{x}_0,\vb*{x}_\tau)}{\ev{\upsilon_t}_\tau}.
\end{equation}

It can be easily verified that
\begin{align}
\sigma_n^y\sigma_m^x-\sigma_n^x\sigma_m^y&\preceq(\sigma_n^z+\mbb{1}_2)\otimes\mbb{1}_2+\mbb{1}_2\otimes(\sigma_m^z+\mbb{1}_2),\\
\sigma_n^x\sigma_m^y-\sigma_n^y\sigma_m^x&\preceq(\sigma_n^z+\mbb{1}_2)\otimes\mbb{1}_2+\mbb{1}_2\otimes(\sigma_m^z+\mbb{1}_2),
\end{align}
where $\mbb{1}_2$ is the two-dimensional identity matrix.
Here, $A\preceq B$ means that $B-A$ is positive semi-definite.
Therefore, $|f_{nm}(t)|$ can be upper bounded as
\begin{equation}\label{eq:fmn.ub}
|f_{nm}(t)|\le\gamma[x_n(t)+x_m(t)].
\end{equation}
Using Eq.~\eqref{eq:fmn.ub}, we can show that the velocity $\upsilon_t$ is bounded from above as
\begin{equation}\label{eq:spin.vt.ub}
\upsilon_t\le\sum_{\ev{n,m}\in\mca{E}}\gamma[x_n(t)+x_m(t)]\le\gamma d_G\|\vb*{x}_t\|_1.
\end{equation}
It should be noted that the total spin $\|\vb*{x}_t\|_1$ is invariant for all times.
For convenience, we define $\mca{M}\coloneqq\|\vb*{x}_t\|_1$.
Using the inequality \eqref{eq:spin.vt.ub} and the topological speed limit \eqref{eq:spin.topo.sl}, we obtain the following bound on the operational time required for transmitting information:
\begin{equation}\label{eq:spin.qsl}
\tau\ge\frac{\mca{W}_1(\vb*{x}_0,\vb*{x}_\tau)}{\gamma d_G\mca{M}}.
\end{equation}
This inequality implies that the speed of information transmission is constrained by the topology of the graph $G$ and the coupling strength.

Now let us discuss a particular case where the graph $G$ is a spin chain of length $N\ge 2$, and the spins $n$ and $n+1$ interact with each other for all $1\le n\le N-1$ \cite{Bose.2003.PRL,Murphy.2010.PRA}. 
Alice prepares the spin chain in the initial state $\varrho_0=\dyad{\varphi_0}$, with spin $1$ in the excited state $\ket{1}$, and all other spins in the ground state $\ket{0}$. Specifically, $\ket{\varphi_0}$ is given by
\begin{equation}
\ket{\varphi_0}=\ket{1}\otimes\ket{0}\otimes\dots\otimes\ket{0}.
\end{equation}
In this setup, the maximum degree of the graph is $d_G=2$, and the total spin is $\mca{M}=1$.
The initial and target vectors are given by $\vb*{x}_0=[1,0,\dots,0]^\top$ and $\vb*{x}_\tau=[0,\dots,0,1]^\top$, respectively.
The Wasserstein distance can also be analytically calculated as $\mca{W}_1(\vb*{x}_0,\vb*{x}_\tau)=N-1$.
Then, according to Eq.~\eqref{eq:spin.qsl}, the minimum time required for transmitting the quantum state is lower bound as
\begin{equation}\label{eq:simp.spin.qsl}
\tau\ge\frac{N-1}{2\gamma}.
\end{equation}
Intriguingly, Eq.~\eqref{eq:simp.spin.qsl} implies that it takes at least a time proportional to the distance between spins to reliably transfer a quantum state.
The longer the distance, the more time is required.
This is in agreement with the numerical result in Ref.~\cite{Murphy.2010.PRA}, wherein an optimal control of the external magnetic field was used.
It is noteworthy that this implication cannot be obtained from the conventional speed limits such as the Mandelstam--Tamm and Margolus--Levitin bounds.

\section{Markovian open quantum systems}
Here we derive a topological speed limit for Markovian open quantum systems.
We consider a finite-dimensional quantum system that is weakly coupled to thermal reservoirs.
The dynamics of the system's reduced density matrix is governed by the local Lindblad equation:
\begin{align}
\dot\varrho_t=-i[H+V_t,\varrho_t]+\sum_k\mca{D}[L_k]\varrho_t,
\end{align}
where $H=\sum_n\epsilon_n\dyad{\epsilon_n}$ is the system Hamiltonian with $\epsilon_n\neq\epsilon_m$ for $n\neq m$, $V_t$ is an external driving field, and $\{L_k\}$ are jump operators that characterize jumps between energy eigenstates with the same energy change $\omega_k$ (i.e., $[L_k,H]=\omega_kL_k$).
We assume that the Hamiltonian has no energy degeneracy.
Each jump operator $L_k$ has a counterpart $L_{k'}$, which corresponds to the reversed jump and satisfies the local detailed balance condition $L_k=e^{s_k/2}L_{k'}^\dagger$.
Here, $ s_k=- s_{k'}$ denotes the change in environmental entropy due to the jump $L_k$. 
We consider the time evolution of the energetic population $x_n(t)=\mel{\epsilon_n}{\varrho_t}{\epsilon_n}$, which can be described by the following equation:
\begin{align}\label{eq:mse.like}
\dot x_n(t)&=-i\mel{\epsilon_n}{[V_t,\varrho_t]}{\epsilon_n}+\sum_k\sum_{m(\neq n)}[r_{nm}^kx_m(t)-r_{mn}^{k'}x_n(t)],
\end{align}
where $r_{mn}^k\coloneqq|\mel{\epsilon_m}{L_k}{\epsilon_n}|^2\ge 0$ is the transition rate satisfying the local detailed balance $r_{mn}^k=e^{ s_k}r_{nm}^{k'}$.
Notice that Eq.~\eqref{eq:mse.like} can be expressed in the form of Eq.~(\DetEq) with
\begin{align}
f_n(t)&=-i\mel{\epsilon_n}{[V_t,\varrho_t]}{\epsilon_n},\\
f_{nm}(t)&=\sum_k[r_{nm}^kx_m(t)-r_{mn}^{k'}x_n(t)].
\end{align}
The graph $G(\mca{V},\mca{E})$ is thus defined by $\mca{V}=\{1,\dots,n,\dots\}$ and $\mca{E}=\{\ev{m,n}\,|\,m<n,\,\exists k~\text{s.t.}~r_{mn}^k\neq 0\}$.
According to Eq.~(\GenSL), we obtain the following speed limit:
\begin{equation}
\tau\ge\frac{\mca{W}_{1,\lambda}(\vb*{x}_0,\vb*{x}_\tau)}{\ev{\upsilon_{t,\lambda}}_\tau}.
\end{equation}

Next, we derive an upper bound for $\upsilon_{t,\lambda}$. First, according to Prop.~\ref{prop:diag.trace.bound}, we have
\begin{equation}\label{eq:Markov.open.tmp1}
\sum_n|f_n(t)|\le\|[V_t,\varrho_t]\|_1\le 2\Delta_{\varrho_t} V_t,
\end{equation}
where $\|X\|_p\coloneqq\tr{|X|^p}^{1/p}$ is the Schatten $p$-norm of operator $X$.
Furthermore, according to Prop.~\ref{prop:cur.ent.dyn}, we can prove that
\begin{equation}\label{eq:Markov.open.tmp2}
\sum_{\ev{m,n}\in\mca{E}}|f_{mn}(t)|\le\frac{\sigma_t}{2}\Phi\qty(\frac{\sigma_t}{2a_t})^{-1}.
\end{equation}
Here, $a_t\coloneqq\sum_k\tr{L_k\varrho_tL_k^\dagger}$ is the dynamical activity rate, $\sigma_t\coloneqq\sigma_t^{\rm pop}+\sigma_t^{\rm env}$ is the total entropy production rate, $\sigma_t^{\rm env}\coloneqq\sum_k\tr{L_k\varrho_tL_k^\dagger} s_k$ is the environmental entropy rate and $\sigma_t^{\rm pop}\coloneqq-\sum_n\dot{x}_n(t)\ln x_n(t)-i\sum_n\mel{\epsilon_n}{[V_t,\varrho_t]}{\epsilon_n}\ln x_n(t)$ is the sum of the Shannon entropy rate of the population distribution and the entropic change contributed by the external Hamiltonian $V_t$.
Combining Eqs.~\eqref{eq:Markov.open.tmp1} and \eqref{eq:Markov.open.tmp2} yields
\begin{equation}
\upsilon_{t,\lambda}\le 2\lambda\Delta_{\varrho_t} V_t + \frac{\sigma_t}{2}\Phi\qty(\frac{\sigma_t}{2a_t})^{-1}.
\end{equation}
Consequently, we obtain the following thermodynamic speed limit:
\begin{equation}\label{eq:Markov.open.qsl1}
\tau\ge\frac{\mca{W}_{1,\lambda}(\vb*{x}_0,\vb*{x}_\tau)}{\ev{2\lambda\Delta_{\varrho_t} V_t + \sigma_t\Phi(\sigma_t/2a_t)^{-1}/2}_\tau}.
\end{equation}

\begin{proposition}\label{prop:diag.trace.bound}
The following inequality holds for arbitrary Hermitian matrix $V$ and density matrix $\varrho$:
\begin{equation}\label{eq:diag.ine}
\sum_n|\mel{\epsilon_n}{[V,\varrho]}{\epsilon_n}|\le\|[V,\varrho]\|_1\le 2\Delta_{\varrho} V.
\end{equation}
\end{proposition}
\begin{proof}
Note that $\mel{\epsilon_n}{[V,\varrho]}{\epsilon_n}$ is pure imaginary.
Let $U=\sum_{n}u_{nn}{\epsilon_n}$ be a diagonal matrix with elements $u_{nn}=\sign({-i\mel{\epsilon_n}{[V,\varrho]}{\epsilon_n}})$.
Here, $\sign(x)=1$ if $x\ge 0$ and $\sign(x)=-1$ if $x<0$.
It is evident that all singular values of matrix $U$ equal $1$; thus, $\|U\|_\infty=1$.
According to von Neumann's trace inequality, the first inequality in Eq.~\eqref{eq:diag.ine} can be proved as
\begin{align}
\sum_n|\mel{\epsilon_n}{[V,\varrho]}{\epsilon_n}|&=|\tr{U[V,\varrho]}|\notag\\
&\le\|U\|_\infty\|[V,\varrho]\|_1\notag\\
&=\|[V_t,\varrho_t]\|_1.
\end{align}

Next, we need only prove the second inequality in Eq.~\eqref{eq:diag.ine}.
To this end, we follow the idea in Ref.~\cite{Funo.2019.NJP}.
Let $\varrho=\sum_ip_i\dyad{i}$ be the spectral decomposition of the density matrix $\varrho$ acting on the $d$-dimensional Hilbert space $\mca{H}$.
Let $\mca{H}'$ be another copy of the Hilbert space $\mca{H}$ with an orthonormal basis $\{\ket{i'}\}$.
Then, $\ket{\varrho}=\sum_i\sqrt{p_i}\ket{i}\otimes\ket{i'}\in\mca{H}\otimes\mca{H}'$ is the purification of the density matrix $\varrho_t$.
We also define $\tilde{V}\coloneqq V\otimes\mbb{1}\in\mca{H}\otimes\mca{H}'$, where $\mbb{1}$ is the identity matrix acting on the Hilbert space $\mca{H}'$.
Let $\Phi(\cdot)=\tr_{\mca{H}'}\qty{\cdot}$ denote the map of the partial trace with respect to $\mca{H}'$.
By simple algebraic calculations, we can verify that
\begin{equation}
\Phi([\tilde{V},\dyad{\varrho}])=[V,\varrho].
\end{equation}
Since the trace norm is contractive under a completely positive and trace-preserving map, we have
\begin{align}
\|[V,\varrho]\|_1&=\|\Phi([\tilde{V},\dyad{\varrho}])\|_1\notag\\
&\le\|[\tilde{V},\dyad{\varrho}]\|_1\notag\\
&=\Delta_{\dyad{\varrho}}\tilde{V}\tr{\sqrt{\dyad{\varrho}+\dyad{\varrho_\perp}}}\notag\\
&=2\Delta_{\dyad{\varrho}}\tilde{V}.\label{eq:ene.fluct.tmp1}
\end{align}
Here, $\ket{\varrho_\perp}$ is a state orthogonal to $\ket{\varrho}$, given by
\begin{equation}
\ket{\varrho_\perp}\coloneqq\frac{(V-\mel{\varrho}{V}{\varrho})\ket{\varrho}}{\Delta_{\dyad{\varrho}}\tilde{V}}.
\end{equation}
In addition, the energy fluctuation $\Delta_{\dyad{\varrho}}\tilde{V}$ can be simplified as follows:
\begin{align}
\Delta_{\dyad{\varrho}}\tilde{V}&=\qty(\mel{\varrho}{\tilde{V}^2}{\varrho}-\mel{\varrho}{\tilde{V}}{\varrho}^2)^{1/2}\notag\\
&=\qty(\tr{V^2\varrho}-\tr{V\varrho}^2)^{1/2}\notag\\
&=\Delta_\varrho V.\label{eq:ene.fluct.tmp2}
\end{align}
Combining Eqs.~\eqref{eq:ene.fluct.tmp1} and \eqref{eq:ene.fluct.tmp2} yields the desired second inequality in Eq.~\eqref{eq:diag.ine}.
\end{proof}

\begin{proposition}\label{prop:cur.ent.dyn}
The following inequality holds:
\begin{equation}\label{eq:prop.cur.ent.dyn.tmp1}
\sum_{\ev{m,n}\in\mca{E}}\qty|\sum_k[r_{nm}^kx_m(t)-r_{mn}^{k'}x_n(t)]|\le\frac{\sigma_t}{2}\Phi\qty(\frac{\sigma_t}{2a_t})^{-1},
\end{equation}
where $\Phi(x)$ is the inverse function of $x\tanh(x)$.
\end{proposition}
\begin{proof}
Note that $[L_k^\dagger L_k,H]=0$ and $\mel{\epsilon_n}{L_k^\dagger L_k}{\epsilon_m}=0$ for $m\neq n$. Therefore,
\begin{align}
\tr{L_k\varrho_tL_k^\dagger}&=\sum_m\mel{\epsilon_m}{\varrho_t}{\epsilon_m}\mel{\epsilon_m}{L_k^\dagger L_k}{\epsilon_m}\notag\\
&=\sum_mx_m(t)\sum_n\mel{\epsilon_m}{L_k^\dagger}{\epsilon_n}\mel{\epsilon_n}{L_k}{\epsilon_m}\notag\\
&=\sum_{m,n}r_{nm}^{k}x_m(t).
\end{align}
We first calculate $\sigma_t^{\rm env}$ as follows:
\begin{align}
\sigma_t^{\rm env}&=\sum_k\tr{L_k\varrho_tL_k^\dagger} s_k\notag\\
&=\frac{1}{2}\sum_k(\tr{L_k\varrho_tL_k^\dagger}-\tr{L_{k'}\varrho_tL_{k'}^\dagger}) s_k\notag\\
&=\frac{1}{2}\sum_k\sum_{m,n}[r_{nm}^{k}x_m(t)-r_{mn}^{k'}x_n(t)] s_k\notag\\
&=\frac{1}{2}\sum_k\sum_{m,n}[r_{nm}^{k}x_m(t)-r_{mn}^{k'}x_n(t)]\ln\frac{r_{nm}^k}{r_{mn}^{k'}}.
\end{align}
Consequently, we can calculate
\begin{align}
\sigma_t&=\sigma_t^{\rm pop}+\sigma_t^{\rm env}\notag\\
&=-\sum_k\sum_{m\neq n}[r_{nm}^kx_m(t)-r_{mn}^{k'}x_n(t)]\ln x_n(t)\notag\\
&+\frac{1}{2}\sum_k\sum_{m,n}[r_{nm}^{k}x_m(t)-r_{mn}^{k'}x_n(t)]\ln\frac{r_{nm}^k}{r_{mn}^{k'}}\notag\\
&=\frac{1}{2}\sum_k\sum_{m,n}[r_{nm}^kx_m(t)-r_{mn}^{k'}x_n(t)]\ln\frac{x_m(t)}{x_n(t)}\notag\\
&+\frac{1}{2}\sum_k\sum_{m,n}[r_{nm}^{k}x_m(t)-r_{mn}^{k'}x_n(t)]\ln\frac{r_{nm}^k}{r_{mn}^{k'}}\notag\\
&=\frac{1}{2}\sum_k\sum_{m,n}[r_{nm}^kx_m(t)-r_{mn}^{k'}x_n(t)]\ln\frac{r_{nm}^kx_m(t)}{r_{mn}^{k'}x_n(t)}\notag\\
&\eqqcolon\sum_k\sum_{m,n}\sigma_{nm}^k(t).
\end{align}
Similarly, the dynamical activity rate can also be calculated as
\begin{align}
a_t&=\sum_k\tr{L_k\varrho_tL_k^\dagger}\notag\\
&=\frac{1}{2}\sum_k\qty[\tr{L_k\varrho_tL_k^\dagger}+\tr{L_{k'}\varrho_tL_{k'}^\dagger}]\notag\\
&=\frac{1}{2}\sum_k\sum_{m,n}\qty[r_{nm}^{k}x_m(t)+r_{mn}^{k'}x_n(t)]\notag\\
&\eqqcolon\sum_k\sum_{m,n}a_{nm}^k(t).
\end{align}
By applying the triangle inequality and Jensen's inequality, we obtain the desired inequality \eqref{eq:prop.cur.ent.dyn.tmp1} as
\begin{align}
&\sum_{\ev{m,n}\in\mca{E}}\qty|\sum_k[r_{nm}^kx_m(t)-r_{mn}^{k'}x_n(t)]|\notag\\
&\le\frac{1}{2}\sum_k\sum_{m,n}|r_{nm}^kx_m(t)-r_{mn}^{k'}x_n(t)|\notag\\
&=\sum_k\sum_{m,n}\frac{\sigma_{nm}^k(t)}{2}\Phi\qty(\frac{\sigma_{nm}^k(t)}{2a_{nm}^k(t)})^{-1}\notag\\
&\le\frac{\sigma_t}{2}\Phi\qty(\frac{\sigma_t}{2a_t})^{-1}.
\end{align}
\end{proof}